\documentclass[journal]{IEEEtran}
\usepackage{cite}
\usepackage{amsmath,amssymb,amsfonts}
\usepackage{algorithmic}
\usepackage{graphicx}
\usepackage{algorithm,algorithmic}
\usepackage{textcomp}
\usepackage{nomencl}
\usepackage{verbatim}
\usepackage{caption}
\usepackage{subcaption}
\usepackage{tabularray}
\usepackage{amsthm}
\usepackage{multirow}
\usepackage{etoolbox}
\newtheoremstyle{mystl}
  {0} 
  {0} 
  {\itshape}
  {}
  {\bfseries}  
  {.} 
  {5pt plus 1pt minus 1pt} 
 {\thmname{#1} \thmnumber{#2}\ifblank{#3}{}{ (\thmnote{#3})}} 
\theoremstyle{mystl}

\newtheorem{thm}{Theorem}
\newtheorem{prop}{Proposition}
\newtheorem{asm}{Assumption}
\newtheorem{rem}{Remark}
\newtheorem{df}{Definition}
\newtheorem{lem}{Lemma}
\def\BibTeX{{\rm B\kern-.05em{\sc i\kern-.025em b}\kern-.08em
    T\kern-.1667em\lower.7ex\hbox{E}\kern-.125emX}}
\makeatletter
\let\NAT@parse\undefined
\makeatother

\usepackage[
    colorlinks=true,
    linkcolor=blue,
    citecolor=blue,
    urlcolor=blue
]{hyperref}
\begin{document}

\title{Grid-Forming Control with Assignable Voltage Regulation Guarantees and Safety-Critical Current Limiting}

\author{Bhathiya Rathnayake and Sijia Geng
\thanks{This work was supported by the Ralph O’Connor Sustainable Energy Institute (ROSEI) at Johns Hopkins University. Corresponding Aurhor: Sijia Geng.}
\thanks{B. Rathnayake is with the Ralph O’Connor Sustainable Energy Institute (ROSEI), Johns Hopkins University, Baltimore, MD, USA (email: brathna1@jh.edu).}
\thanks{S. Geng is with the Department of Electrical and Computer Engineering, Johns Hopkins University, Baltimore, MD, USA (email: sgeng@jhu.edu).}
}

\maketitle

\begin{abstract} This paper develops a nonlinear grid-forming (GFM) controller with provable voltage-formation guarantees, with over-current limiting enforced via a control-barrier-function (CBF)-based safety filter. The nominal controller follows a droop-based inner–outer architecture, in which voltage references and frequency are generated by droop laws, an outer-loop voltage controller produces current references using backstepping (BS), and an inner-loop current controller synthesizes the terminal voltage. The grid voltage is treated as an unknown bounded disturbance, without requiring knowledge of its bound, and the controller design does not rely on any network parameters beyond the point of common coupling (PCC). To robustify voltage formation against the grid voltage, a deadzone-adapted disturbance suppression (DADS) framework is incorporated, yielding practical voltage regulation characterized by asymptotic convergence of the PCC voltage errors to an \textit{assignably} small and known residual set. Furthermore, the closed-loop system is proven to be globally well posed, with all physical and adaptive states bounded and voltage error transients (due to initial conditions) decaying exponentially at an assignable rate. On top of the nominal controller, hard over-current protection is achieved through a minimally invasive CBF-based safety filter that enforces strict current limits via a single-constraint quadratic program. The safety filter is compatible with any locally Lipschitz nominal controller. Rigorous analysis establishes forward invariance of the safe-current set and boundedness of all states under current limiting. Numerical results demonstrate improved transient performance and faster recovery during current-limiting events when the proposed DADS-BS controller is used as the nominal control law, compared with conventional PI-based GFM control.
\end{abstract}


\section{Introduction}

\IEEEPARstart{T}{he} electric power system is undergoing a structural transition driven by the large-scale integration of inverter-based resources (IBRs), including solar photovoltaics, battery energy storage systems, and power-electronic-based wind generation. As conventional synchronous generators (SGs) are retired, the grid is gradually losing the electromechanical characteristics that historically provided inertia, frequency support, and fault current \cite{rathnayake2021grid}. This transition introduces several  challenges: reduced inertia leads to faster and more sensitive frequency dynamics \cite{milano2018foundations}; limited short-circuit current undermines traditional protection schemes \cite{baeckeland2024overcurrent}; and interactions among inverter controllers can induce complex nonlinear phenomena \cite{chatterjee2024sensitivity,chatterjee2025effects}.

Grid-forming (GFM) inverters are widely regarded as vital for future low-inertia power systems due to their capability in directly regulating terminal voltage magnitude and frequency, enabling virtual inertia, frequency regulation, black-start, and robust operation in weak or islanded grids \cite{lin2020research}.
A large body of work has focused on GFM controllers, including, \textit{droop-based controllers} \cite{tayab2017review,geng2022unified}, \textit{virtual synchronous machines} (VSMs) \cite{beck2007virtual,alatrash2012generator,zhong2010synchronverters}, and \textit{virtual oscillator controllers} (VOCs) with their \textit{dispatchable (d)} variants  \cite{johnson2015synthesizing,sinha2015uncovering,colombino2019global,lu2019grid}. Droop control, originating from early microgrid research \cite{chandorkar1993control}, enforces steady-state  active-power--frequency and reactive-power--voltage relationships that enable synchronization and power sharing. VSMs emulate SG dynamics within the inverter controller, while VOCs regulate voltage and frequency through digitally implemented nonlinear oscillators. Despite architectural differences, these GFM controllers present a common external behavior: the inverter acts as a controllable voltage source whose amplitude and frequency vary with reactive and active power, respectively, allowing fast voltage regulation and frequency support.

Unlike SGs, GFM inverters possess limited over-current capability due to semiconductor constraints, typically tolerating only $1.2$--$2$ p.u. current, whereas SGs can, in general, supply 5–7 p.u. current \cite{fan2022review}. Consequently, effective current-limiting is critical for device protection during disturbances. Existing approaches are commonly categorized as \emph{direct} or \emph{indirect} methods \cite{baeckeland2024overcurrent}. Direct methods, such as current-reference saturation \cite{mahamedi2018sequence,sadeghkhani2016current} and switch-level limiting \cite{gurule2019grid,liu2021current}, enforce current bounds accurately but may compromise voltage-source behavior and post-fault recovery. Indirect methods—including virtual impedance \cite{paquette2014virtual,wang2014virtual} and voltage-based limiting \cite{erckrath2022voltage,chen2020use}—better preserve voltage-source characteristics but generally lack guarantees of hard current bounds. Constraint-aware voltage shaping methods have also been proposed to indirectly limit current through feasibility projections or optimization \cite{desai2024saturation,gross2025constraint}. More recently, \textit{control barrier function (CBF)}-based approaches have emerged as a systematic framework for enforcing safety constraints in control systems \cite{ames2016control} and power systems \cite{kundu2019distributed}. In the GFM context, CBF-based current limiters \cite{schneeberger2025safety,joswig2024safe} explicitly impose hard current bounds through online optimization, typically via a \textit{quadratic program (QP)}, while minimally modifying a nominal voltage-forming controller.

Despite these advances, several challenges remain. First, many representative GFM controllers across different architectures \cite{beck2007virtual,alatrash2012generator,zhong2010synchronverters,tayyebi2022grid,chandorkar1993control} rely on proportional–integral–derivative (PID) control or variants applied to inherently nonlinear inverter dynamics. In such designs, theoretical guarantees are typically established separately from controller synthesis. While model-based nonlinear designs aim to more explicitly link control laws with guarantees, many existing approaches \cite{karunaratne2023nonlinear,lourenco2024nonlinear,panda2023extended} assume accurate network models, require observers or estimators, or rely on higher-order derivatives of electrical measurements, which can be difficult to obtain and tune in practice. Second, the systematic integration of hard over-current protection into GFM controllers remains largely unresolved. Current limiting is often implemented using heuristic mechanisms \cite{fan2022review,baeckeland2024overcurrent}, and it is unclear whether such mechanisms can be non-trivially incorporated into existing GFM control architectures. Although CBFs provide a principled tool for embedding current constraints, only a limited number of works \cite{schneeberger2025safety,joswig2024safe, kundu2019distributed} have explored this direction for GFM inverters.

To address such issues, this paper develops a droop-based, backstepping (BS) GFM controller for grid-connected IBRs. The grid voltage is treated as an unknown bounded disturbance,  without requiring knowledge of the bound. Furthermore, the design requires no knowledge of network parameters beyond the point of common coupling (PCC). 

To robustify the voltage-forming objective against the grid voltage, we tailor the recently introduced \textit{deadzone-adapted disturbance suppression (DADS)} framework \cite{karafyllis2024deadzone,karafyllis2025deadzone,karafyllis2025robust} to the inverter $\rm dq$-channel dynamics and the active/reactive power filters. DADS combines (i) a deadzone in the adaptive law for a controller gain— activating adaptation only when a suitable error-energy measure is above a prescribed threshold, thereby preventing gain drift—with (ii) dynamic nonlinear damping terms that inject additional dissipation. This structure yields disturbance suppression that asymptotically regulates the relevant states to an \textit{assignably small and known} residual set in the presence of unknown but constant network parameters and bounded disturbances of arbitrarily large magnitude.

On top of the nominal voltage-forming controller, we incorporate a QP-based safety filter that enforces a hard bound on the inverter terminal current via a CBF constraint. The filter is minimally invasive: it modifies the nominal terminal-voltage command only when necessary to prevent current-limit violations,  thereby preserving GFM behavior as closely as possible within the safe region.

The contributions of this work are summarized as follows.

\textit{1) Nominal Droop-based BS GFM controller with DADS (DADS-BS) (Section~\ref{secDADScntrol}):} We propose a nonlinear GFM controller with an
inner–outer architecture in which the  frequency and PCC voltage references are generated by active/reactive power
droop with second-order power filters. The outer-loop voltage controller produces reference currents that act as \emph{virtual} controls in  BS, while the inner-loop current controller synthesizes the terminal-voltage inputs and incorporates the DADS framework. The design requires no knowledge of network parameters beyond the PCC and employs neither disturbance observers nor parameter estimators.

\textit{2) Global well-posedness and boundedness under bounded disturbances (Theorem~\ref{thm1}):}
For the nominal DADS-BS (without current limiting), we establish global existence and uniqueness of solutions and boundedness of \emph{all} closed-loop states--the PCC voltage, inverter and grid currents, power-filter states, and DADS adaptive gains. The closed-loop system satisfies the bounded-input bounded-state property with respect to the bounded grid voltage.

\textit{3) Transient performance and asymptotic voltage-formation guarantees (Theorem~\ref{thm2}):}
We show that convergence of the PCC voltage errors to a region determined by unknown constant network parameters and grid voltage bounds is \emph{exponential} with an \emph{assignable} convergence rate. Moreover, we establish practical voltage regulation: the PCC voltage asymptotically tracks the droop-generated reference within a \emph{designer-specified} residual neighborhood, independently of the disturbance magnitude and unknown constant network parameters. These guarantees remain valid even under severe fault conditions, such as three-phase-to-ground faults at the grid. Explicit tradeoffs among residual size, adaptation rates, and control effort are discussed.

\textit{4) Hard over-current protection with closed-loop boundedness guarantees (Theorems~\ref{thm4} and~\ref{thm6}):}
We enforce strict terminal-current limits using a CBF-based safety filter that minimally modifies the nominal terminal-voltage command through a single-constraint QP. A closed-form expression for the safety-filtered control is obtained, eliminating the need for online QP solvers. For arbitrary locally Lipschitz nominal controllers filtered through this safety filter, we prove forward invariance of the safe-current set and boundedness of all remaining physical states. Specializing to the proposed DADS-BS, we further analyze the interaction between current-limiting events and adaptive dynamics, and establish boundedness of the adaptive gains under a mild activation-pattern assumption on the safety filter (finite number and total duration of ON intervals). Numerical results demonstrate improved transient behavior and recovery of nominal GFM behavior during current-limiting events when using DADS-BS as the nominal controller, compared with standard PI-based GFM control.

\noindent \textbf{Notation.} $\mathbb{R}_{+} := [0,+\infty)$ and $\mathbb{R}_{>0} := (0,+\infty)$. For a vector $\boldsymbol x \in \mathbb{R}^{n_1}$, $\Vert \boldsymbol x\Vert $ denotes its Euclidean norm, and for $x\in\mathbb{R}$, $\vert x\vert = \Vert x\Vert$.  Let $D \subset \mathbb{R}^{n_1}$ be an open set and let $S \subset \mathbb{R}^{n_1}$ be a set that satisfies
$D \subseteq S \subseteq  \rm{cl}(\it D)$, where $\rm{cl}(\it D)$ is the closure of $D$. By $C^{0}(S;\Omega)$,
we denote the class of continuous functions on $S$, which take values in
$\Omega \subseteq \mathbb{R}^{n_2}$. Let the non-empty set $\Theta \subset \mathbb{R}^{n_3}$ be given. By
$L^{\infty}(\mathbb{R}_{+};\Theta)$ we denote the class of essentially bounded,
Lebesgue measurable functions $\boldsymbol{d}:\mathbb{R}_{+}\rightarrow \Theta$. For
$\boldsymbol{d} \in L^{\infty}(\mathbb{R}_{+};\Theta)$, we define $
\Vert \boldsymbol{d}\Vert_{\infty} := \operatorname*{sup}_{t \geq 0}(\Vert  \boldsymbol{d}(t)\Vert )$, where $\operatorname*{sup}_{t \geq 0}(\Vert  \boldsymbol{d}(t)\Vert )$ is the essential supremum. By $\mathcal K$, we denote the class of increasing continuous functions
$\alpha:\mathbb{R}_{+}\rightarrow \mathbb{R}_{+}$ with $\alpha(0)=0$.

The paper is organized as follows. Section \ref{secII} presents generic modeling of grid-connected IBRs. Section \ref{secDADScntrol} introduces the DADS-BS framework for droop-based GFM control. Section \ref{secCBF} proposes a CBF-based QP safety filter for enforcing strict terminal-current constraints. Section \ref{secSim} presents numerical simulations. Section \ref{secProof} provides the proofs of all theoretical results presented in Sections \ref{secDADScntrol} and \ref{secCBF}. Section \ref{secCon} concludes the paper.

\section{Modeling of Grid-Connected IBRs}\label{secII}

\subsection{Reference Frames and Coordinate Transformations}\label{subsec:frames}
Figure \ref{GFM_inftybus} shows a three-phase voltage source converter (VSC) connected to the grid through an output filter and a transmission line. The inverter synthesizes terminal voltages $v_{\rm t,abc} := (v_{\rm t,a},v_{\rm t,b},v_{\rm t,c})$, which drive terminal currents $i_{\rm t,abc} := (i_{\rm t,a},i_{\rm t,b},i_{\rm t,c})$. The point of common coupling (PCC) voltages are denoted by $v_{\rm c,abc} := (v_{\rm c,a},v_{\rm c,b},v_{\rm c,c})$, and the currents by $i_{\rm g,abc} := (i_{\rm g,a},i_{\rm g,b},i_{\rm g,c})$. The grid-side voltages are $v_{\rm g,abc} := (v_{\rm g,a},v_{\rm g,b},v_{\rm g,c})$. 

\begin{figure}
    \centering
    \includegraphics[width=1\linewidth]{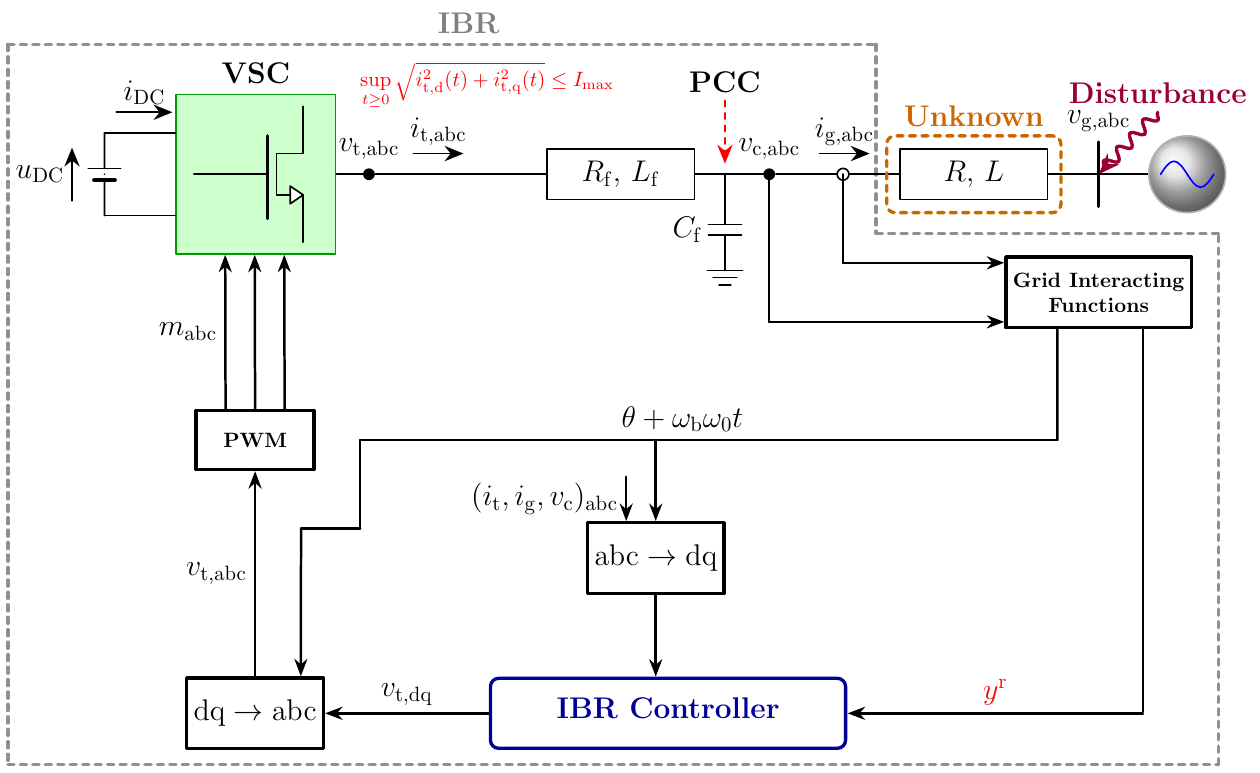}
\caption{An IBR connected to the grid. In GFL operation, the objective is to inject prescribed active and reactive power at the PCC; in GFM operation, the objective is to establish the voltage waveform at the PCC.}

    \label{GFM_inftybus}
\end{figure}

Although $\rm abc$ variables provide a direct physical description, they remain sinusoidal and time-varying even under steady-state balanced conditions, which complicates analysis and control design. Reference-frame transformations are therefore employed to express these quantities in rotating coordinates, where steady-state signals appear constant.

Following standard practice in IBR modeling \cite{geng2022unified,chatterjee2025effects}, we use two reference frames: a \textit{global} rotating $\rm DQ$ frame for network-level modeling and stability analysis, and \textit{local} rotating $\rm dq$ frames attached to individual inverters for control design and implementation.

Each inverter operates in a local $\rm dq$ frame rotating at angular velocity $\omega\,\omega_{\rm b}$ rad/s, where $\omega$ is the per-unit electrical frequency and $\omega_{\rm b}$ is the base electrical angular frequency (e.g., $\omega_{\rm{b}}=2\pi\times 60$ rad/s). Under balanced three-phase operation, any set of $\rm abc$ variables can be mapped without loss of information to two components, $\rm d$ and $\rm q$, via the Park transformation (see \cite{schiffer2016survey}). A global rotating $\rm DQ$ frame is introduced to provide a unified coordinate system for the network. This frame rotates at angular velocity $\omega_{\rm DQ}\omega_{\rm b}$ rad/s and is typically chosen such that $\omega_{\rm DQ}=\omega_0$, the steady-state grid frequency, so that electrical quantities are constant at equilibrium. Figure \ref{DQdqframes} shows the PCC voltage expressed in both frames. In the local $\rm dq$ frame, the PCC voltage is represented as the complex quantity
\begin{align*}
    v_{\rm c} = v_{\rm c,d} + j v_{\rm c,q},
\end{align*}
while in the global $\rm DQ$ frame it is given by
\begin{align*}
    v_{\rm c} = v_{\rm c,D} + j v_{\rm c,Q}.
\end{align*}
All other voltages and currents are represented analogously.

Let $\theta$ denote the angular displacement of the local $\rm dq$ frame relative to the global $\rm DQ$ frame, as illustrated in Fig. \ref{DQdqframes}. The evolution of this angle satisfies
\begin{align*}
    \dot{\theta} = \omega_{\rm b} (\omega - \omega_0),
\end{align*}
capturing the relative frequency deviation between the frames.

The transformation between the $\rm dq$ and $\rm DQ$ coordinates is given by the \textit{invertible} rotation matrix
\begin{align*}
\boldsymbol R(\theta) :=
\begin{bmatrix}
\cos\theta & -\sin\theta\\
\sin\theta & \cos\theta
\end{bmatrix},
\end{align*}
so that, for example, the PCC voltage components satisfy
\begin{align*}
\begin{bmatrix}
v_{\rm c,D}\\
v_{\rm c,Q}
\end{bmatrix}
= \boldsymbol R(\theta)
\begin{bmatrix}
v_{\rm c,d}\\
v_{\rm c,q}
\end{bmatrix}.
\end{align*}
Transformations for all other quantities follow analogously. In the single-inverter infinite-bus setting considered below, all inverter states and control laws are expressed in the local $\rm dq$ frame. The grid voltage is assumed bounded; hence, no explicit $\rm DQ$--$\rm dq$ transformation is required in the analysis. Such a transformation is only used in simulations to convert the $\rm DQ$ grid voltage into the local $\rm dq$ frame. 
\begin{figure}
    \centering
    \includegraphics[width=0.8\linewidth]{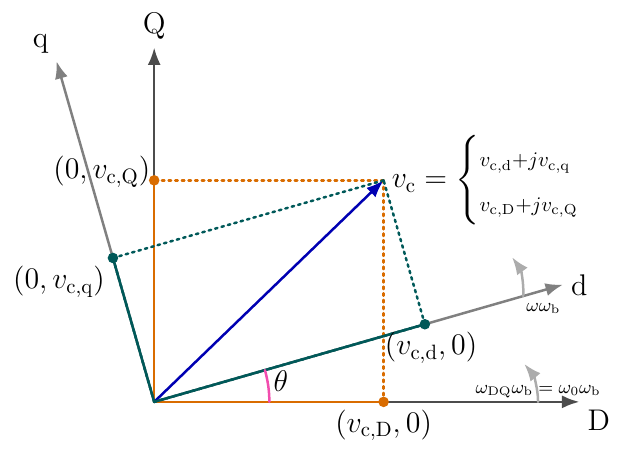}
    \caption{PCC voltage in the $\rm dq$ and $\rm DQ$ frames.}
    \label{DQdqframes}
\end{figure}

\begin{asm}
The DC-side dynamics are neglected, with the DC-link voltage $u_{\rm DC}$ assumed constant and the DC current $i_{\rm DC}$ determined algebraically to supply the required power. PWM switching dynamics are assumed sufficiently fast so that the inverter terminal voltages track their commanded $\rm dq$ setpoints instantaneously. The transmission line is modeled by lumped series resistance and inductance.
\end{asm}
\subsection{Circuit Dynamics}
The $\rm d$ and $\rm q$ channel circuit dynamics of an IBR connected to the grid, as shown in Fig.~\ref{GFM_inftybus}, are given by
\begin{align}\label{inv1}
    \dot{v}_{\rm{c,d}} &= \omega_{\rm b} \omega  v_{\rm{c,q}} + \frac{\omega_{\rm b}}{C_{\rm f}}\big(i_{\rm{t,d}} - i_{\rm{g,d}}\big), \\ \label{inv2}
    \dot{v}_{\rm{c,q}} &= -\omega_{\rm b} \omega v_{\rm{c,d}} + \frac{\omega_{\rm b}}{C_{\rm f}}\big(i_{\rm{t,q}} - i_{\rm{g,q}}\big), \\ \label{itddn}
    \dot{i}_{\rm{t,d}} &= \omega_{\rm b} \omega i_{\rm{t,q}} + \frac{\omega_{\rm b}}{L_{\rm f}}\big(v_{\rm{t,d}} - v_{\rm{c,d}}\big) - \frac{\omega_{\rm b}R_{\rm f} }{L_{\rm f}} i_{\rm{t,d}}, \\ \label{itqdn}
    \dot{i}_{\rm{t,q}} &= -\omega_{\rm b} \omega i_{\rm{t,d}} + \frac{\omega_{\rm b}}{L_{\rm f}}\big(v_{\rm{t,q}} - v_{\rm{c,q}}\big) - \frac{\omega_{\rm b} R_{\rm f}}{L_{\rm f}} i_{\rm{t,q}},\\\label{invend1} 
    \dot{i}_{\rm{g,d}} &= \omega_{\rm b}\omega i_{\rm{g,q}}+\frac{\omega_{\rm b}}{L}\big(v_{\rm{c,d}}-v_{\rm{g,d}}\big)-\frac{\omega_{\rm b} R}{L} i_{\rm{g,d}},\\\label{invend}
    \dot{i}_{\rm{g,q}} &= -\omega_{\rm b}\omega i_{\rm{g,d}}+\frac{\omega_{\rm b}}{L}\big(v_{\rm{c,q}}-v_{\rm{g,q}}\big)-\frac{\omega_{\rm b}R}{L} i_{\rm{g,q}}.
\end{align}
See Appendix A of \cite{geng2022unified} for a derivation of IBR dynamics in a rotating reference frame.

Define the inverter electrical state vector
\begin{align*}
    \boldsymbol x_{\text{inv}} := \big[v_{\rm {c,d}}, \, v_{\rm {c,q}}, \,  i_{\rm {t,d}}, \, i_{\rm {t,q}}, \,  i_{\rm {g,d}}, \,  i_{\rm {g,q}}\big]^{  \top} \in \mathbb{R}^6,
\end{align*}
with initial condition $ \boldsymbol x_{\text{inv}}(0) \in \mathbb{R}^6$. The signals $v_{\rm {t,d}}, v_{\rm {t,q}}$ respectively denote the $\rm d$- and $\rm q$-axis control inputs, to be specified by an inner controller. The signals $v_{\rm {g,d}}, v_{\rm {g,q}}$ represent the $\rm d$– and $\rm q$–components of the grid voltage. The electrical frequency $\omega$ is locked to the grid in GFL operation, for example via a PLL, whereas in GFM operation, it is autonomously generated by the inverter, for example, using droop. For notational convenience, we collect the constant physical parameters into
\begin{align}\label{theta_def}
\boldsymbol \Omega
:=
[\omega_{\rm b},\, C_{\rm f},\, L_{\rm f},\, R_{\rm f},\, L,\, R]^\top \in\mathbb{R}_{>0}^{6}.
\end{align}

The network parameters $R,L>0$ belong to the external network (see Fig. \ref{GFM_inftybus}) and are typically unknown to the inverter control design \textit{a priori} or subject to uncertainty. Therefore, throughout the analysis, we treat $R$ and $L$ as unknown but positive constants. No knowledge of $R$ or $L$ is required for the controller design.

\begin{asm}\label{assdis}
    The grid voltage  $v_{\rm g} = v_{\rm g,d} + jv_{\rm g,q}$ is a  bounded signal such that
    \begin{align*}
        v_{\rm g,d},v_{\rm g,q} \in L^{\infty}(\mathbb{R}_{+};\mathbb{R}).
    \end{align*}
\end{asm}
Although we assume that the grid voltage is bounded, no numerical upper bound is used in the controller design. 

In practical implementations, a hard current magnitude constraint 
$\sqrt{i_{\rm g,d}^2(t) + i_{\rm g,q}^2}(t) \le I_{\max}$ for all $t\geq 0$ must be enforced to reflect semiconductor and thermal limits. This constraint will be explicitly enforced through the proposed safety-critical control design discussed in Section~\ref{secCBF}.

\section{Grid-Forming Inverter Control}
\label{secDADScntrol}

This section presents a nonlinear GFM controller obtained by integrating droop laws with a deadzone-adapted disturbance suppression (DADS)-type backstepping (BS) design. The primary objectives are:  
1) to regulate the PCC voltage components $(v_{\rm c,d}, v_{\rm c,q})$ to the droop-generated references $(v^{\rm r}_{\rm c,d}, v^{\rm r}_{\rm c,q})$, and  
2) to ensure boundedness of all closed-loop states, including terminal and grid currents, the power filter states, and the adaptive gains. The proposed controller achieves these objectives despite the presence of unknown but bounded grid voltage and unknown  constant network parameters $R,L>0$. All control laws are designed in the local $\rm dq$ reference frame. Current limiting on $i_{\rm t}$ is not considered in this section, and will be addressed in Section \ref{secCBF}.

\subsection{Droop Relationships}\label{secDroop}

\begin{figure}
    \centering
    \includegraphics[width=1\linewidth]{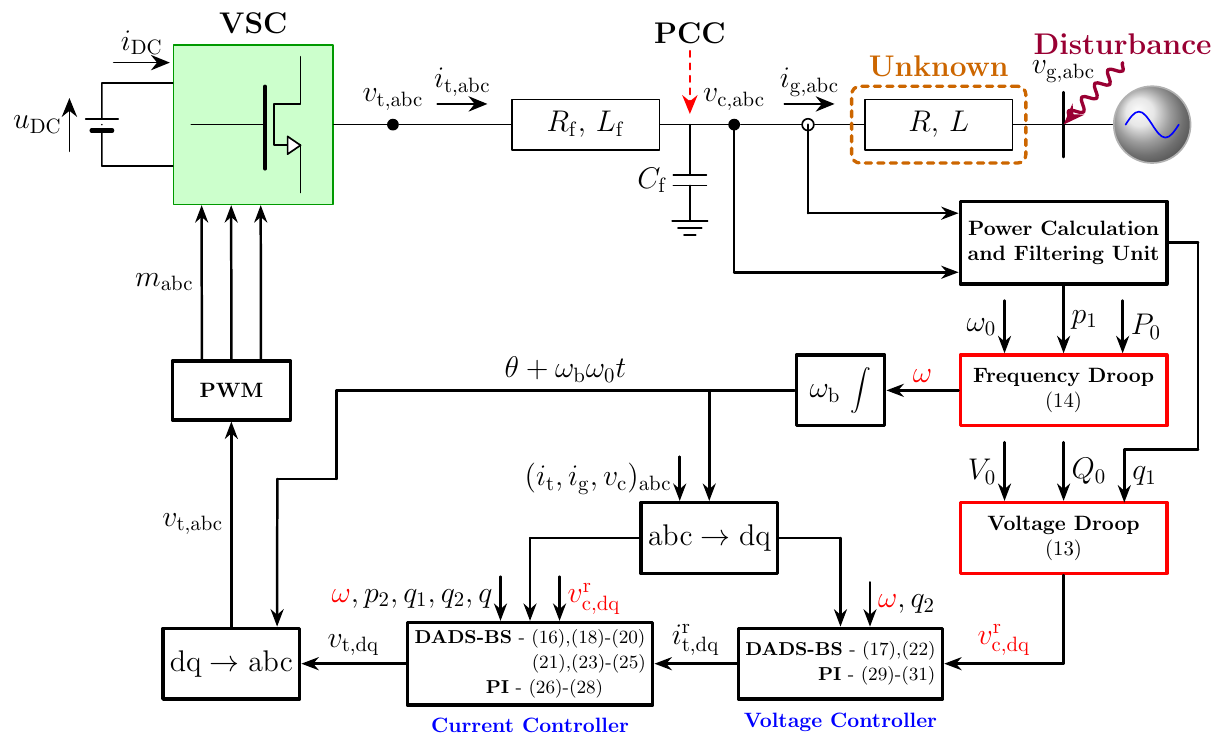}
    \caption{Droop-based GFM control architecture with DADS-BS or PI lower-level control loops.}
    \label{drooparch}
\end{figure}

We first define the instantaneous active and reactive powers at the PCC as
\begin{align}\label{inst_act}
    p &= v_{\rm{c,d}}i_{\rm{g,d}}+v_{\rm{c,q}}i_{\rm{g,q}},\\ \label{inst_react}
    q &= v_{\rm{c,q}}i_{\rm{g,d}}-v_{\rm{c,d}}i_{\rm{g,q}}.
\end{align}
The instantaneous powers $(p,q)$ are passed through second-order low-pass filters with saturation. Specifically, the filtered reactive and active powers evolve according to
\begin{align}
    \dot{q}_1 &= q_2, \qquad
    \dot{q}_2 = -2\xi_{\rm q}\omega_{\rm{qc}}q_2 - \omega_{\rm{qc}}^2\big(q_1 - \operatorname{sat}_{\overline Q}(q)\big), \label{Q_filter}\\
    \dot{p}_1 &= p_2, \qquad
    \dot{p}_2 = -2\xi_{\rm p}\omega_{\rm {pc}}p_2 - \omega_{\rm {pc}}^2\big(p_1 - \operatorname{sat}_{\overline P}(p)\big), \label{P_filter}
\end{align}
with initial conditions $q_1(0), q_2(0), p_1(0), p_2(0) \in \mathbb{R}$. The filter parameters satisfy
$\xi_{\rm q}, \xi_{\rm p} > 1$ and $\omega_{\rm qc}, \omega_{\rm pc} > 0$, while
$\overline Q, \overline P \in  (0,+\infty]$ denote saturation thresholds. The saturation function $\operatorname{sat}_{\overline x}:\mathbb{R}\to\mathbb{R}$ with $\overline{x}\in (0,+\infty]$ is defined as
\begin{align}\label{sat}
    \operatorname{sat}_{\overline x}(y) :=
    \begin{cases}
        \overline x, & y > \overline x,\\
        y, & -\overline x \le y \le \overline x,\\
        -\overline x, & y < -\overline x.
    \end{cases}
\end{align}
Based on the filtered powers $(p_1,q_1)$, the GFM voltage reference and frequency are generated via droop relationships. In particular, the reactive-power–voltage droop law is
\begin{align}
    v^{\rm r}_{\rm{c,d}} &= V_0 + K_{\rm Q}\big(Q_{0} - q_1\big), \qquad
    v^{\rm r}_{\rm{c,q}} \equiv 0, \label{vcd_star}
\end{align}
while the active-power–frequency droop law is given by
\begin{align}
    \omega &= \omega_0 + K_{\rm P}\big(P_{0} - p_1\big).\label{omega_droop}
\end{align}
Here, $V_0, \omega_0 \in \mathbb{R}$ denote nominal voltage magnitude and frequency in p.u.,
$K_{\rm Q}, K_{\rm P} \in \mathbb{R}$ are droop gains in p.u., and
$P_{0}, Q_{0} \in \mathbb{R}$ are active- and reactive-power setpoints in p.u. The references $(v^{\rm r}_{\rm c,d}, v^{\rm r}_{\rm c,q})$ are then tracked by an inner-loop controller, which generates the terminal voltage inputs $(v_{\rm t,d}, v_{\rm t,q})$. Figure \ref{drooparch} shows the droop-based control architecture considered below. 

We collect the inverter electrical states $\boldsymbol{x}_{\rm inv}$ and the power-filter states into the augmented state vector
\begin{align}\label{xstte}
    \boldsymbol x := \big[ \boldsymbol x_{\text{inv}}, \, q_1 ,\,  q_2, \, p_1,\, p_2\big]^{  \top} \in \mathbb{R}^{10}.
\end{align}
\begin{rem}\rm 
The inclusion of the saturation operator in the reactive-power--voltage droop channel is instrumental for ensuring bounded droop-generated voltage references and for capturing the physical limitations of GFM inverters. In particular, the droop law \eqref{vcd_star} must generate bounded voltage references; without saturation, the voltage reference may become unbounded under large reactive-power deviations. In practice, GFM inverters have finite apparent-power ratings and voltage levels, and exceeding these limits can lead to thermal overstress and malfunction of semiconductor devices and violations of grid codes. The operator $\operatorname{sat}_{\overline Q}(\cdot)$ therefore provides a safe and physically meaningful operating envelope for the reactive-power--voltage droop.

From a theoretical standpoint, a saturation operator on the active-power--frequency droop is not required for the boundedness claims (Theorem \ref{thm1}) or the performance guarantees (Theorem \ref{thm2}). One may therefore set $\overline P = +\infty$ without affecting the analysis. Imposing saturation solely on the reactive-power--voltage channel is sufficient, since it enforces boundedness of the droop-generated voltage reference, which—together with the controller introduced below—yields bounded closed-loop behavior, including the electrical frequency.
\end{rem}
\subsection{DADS-BS Control Laws}
The proposed controller follows a typical inner-outer loop control architecture. The PCC voltage dynamics are shaped by introducing reference terminal currents $i^{\rm r}_{\rm t,d}$ and $i^{\rm r}_{\rm t,q}$, which act as \textit{virtual control inputs} in backstepping (BS) and serve as \textit{outer-loop voltage controllers}. The actual control inputs, namely the terminal voltages $v_{\rm t,d}$ and $v_{\rm t,q}$, are then designed to track these reference currents while canceling matched nonlinearities, thereby serving as \textit{inner-loop current controllers}.

To robustify the design against disturbances, the BS control is augmented with DADS-type nonlinear damping terms containing dynamically adapted gains  $z_{\rm d}$ and $z_{\rm q}$. In particular, the dynamic gains modulate the damping injection and are updated through deadzone-based adaptation laws.

We now present the resulting control laws.

\noindent \underline{$\rm d$-axis Control Law}

\begin{align}\label{dads_d1}
\begin{split}
       v_{\rm{t,d}} &= \frac{L_{\rm f}}{\omega_{\rm b}}\bigg[-2\omega_{\rm b} \omega \big(i_{\rm{t,q}} - i_{\rm{g,q}}\big)+\frac{ \omega_{\rm b} R_{\rm f}}{L_{\rm f}} i_{\rm{t,d}} \\&\quad\quad\quad+ \omega_{\rm b}\bigg(\frac{1}{L_{\rm f}}+\omega^2 C_{\rm f} \bigg) v_{\rm{c,d}} + C_{\rm f} K_{\rm P} p_{\rm 2}v_{\rm{c,q}}\\
       &\quad\quad\quad -K_{\rm VC}\big( i_{\rm{t,d}}-i^{\rm r}_{\rm{t,d}}\big)+\frac{C_{\rm f}K_{\rm VC}^2}{\omega_{\rm b}}\big( v_{\rm{c,d}}-v^{\rm r}_{\rm{c,d}}\big)\\&\quad\quad\quad+\frac{2\xi_{\rm q}\omega_{\rm qc}K_{\rm Q}C_{\rm f}}{\omega_{\rm b}}q_2+\frac{\omega^2_{\rm qc}K_{\rm Q}C_{\rm f}}{\omega_{\rm b}}\big(q_1 - \operatorname{sat}_{\overline Q}(q)\big)\\&\quad\quad\quad+u_{\rm d}\bigg], \hspace{10pt} K_{\rm VC}>0,
\end{split}\\\label{itd_r}
i_{\rm{t,d}}^{{\rm r}} &= i_{\rm{g,d}}-C_{\rm f}\omega v_{\rm{c,q}}-\frac{C_{\rm f}K_{\rm Q}}{\omega_{\rm b}}q_{2}-\frac{C_{\rm f}K_{\rm VC}}{\omega_{\rm b}}\big(v_{\rm{c,d}}-v^{\rm r}_{\rm{c,d}}\big),\\
    \label{u_d}
    \begin{split}
u_{\rm d} &= -\Big(K_{\rm CC}+\frac{(1+e^{z_{\rm d}})\omega_{\rm b}^2}{4\mu_{\rm d}}\big(1+i^2_{\rm{g,d}}+v^2_{\rm{c,d}}\big)\Big)\big( i_{\rm{t,d}}-i^{\rm r}_{\rm{t,d}}\big)\\&\quad-\frac{\omega_{\rm b}}{C_{\rm f}}\big( v_{\rm{c,d}}-v^{\rm r}_{\rm{c,d}}\big), \hspace{10pt}K_{\rm CC},\mu_{\rm d}>0,
\end{split}\\
    \dot{z}_{\rm d} &= \Gamma_{\rm d} e^{-z_{\rm d}}\max\Big\{W_{\rm d}-\varepsilon,0\Big\}, \hspace{10pt} \Gamma_{\rm d},\varepsilon >0,\label{z_ddot}\\
        \label{dads_dend}
W_{\rm d} &= \frac{1}{2}\big(v_{\rm{c,d}} - v^{\rm r}_{\rm{c,d}}\big)^2+\frac{1}{2}\big(i_{\rm{t,d}} - i^{\rm r}_{\rm{t,d}}\big)^2.
\end{align}

\noindent \underline{$\rm q$-axis Control Law}

\begin{align}
    \label{dads_q1}
    \begin{split}
        v_{\rm{t,q}} &= \frac{L_{\rm f}}{\omega_{\rm b}}\bigg[2\omega_{\rm b}\omega \big(i_{\rm{t,d}}-i_{\rm{g,d}}\big)+ \frac{\omega_{\rm b} R_{\rm f}}{L_{\rm f}} i_{\rm{t,q}} \\&\quad\quad\quad+ \omega_{\rm b}\bigg(\frac{1}{L_{\rm f}}+C_{\rm f} \omega^2+\frac{C_{\rm f}K^2_{\rm VC}}{\omega_{\rm b}^2} \bigg)v_{\rm{c,q}}-C_{\rm f} K_{\rm P} p_{\rm 2}v_{\rm{c,d}}\\
        &\quad\quad\quad -K_{\rm VC} \big(i_{\rm{t,q}} - i^{\rm r}_{\rm{t,q}}\big)+u_{\rm q}\bigg],\hspace{10pt} K_{\rm VC}>0,
    \end{split}\\\label{itq_r}
         i^{\rm r}_{\rm{t,q}} &= i_{\rm{g,q}}+C_{\rm f}\omega v_{\rm{c,d}}-\frac{C_{\rm f}K_{\rm VC}}{\omega_{\rm b}}v_{\rm{c,q}},\\\label{u_q}
         \begin{split}
    u_{\rm q} &= -\Big(K_{\rm CC}+\frac{(1+e^{z_{\rm q}})\omega_{\rm b}^2}{4\mu_{\rm q}}\big(1+i^2_{\rm{g,q}}+v^2_{\rm{c,q}}\big)\Big)\big( i_{\rm{t,q}}-i^{\rm r}_{\rm{t,q}}\big)\\&\quad-\frac{\omega_{\rm b}}{C_{\rm f}}v_{\rm{c,q}}, \hspace{10pt} K_{\rm CC},\mu_{\rm q}>0,
\end{split}\\\label{z_qdot}
    \dot{z}_{\rm q} &= \Gamma_{\rm q} e^{-z_{\rm q}}\max\Big\{W_{\rm q}-\varepsilon,0\Big\}, \hspace{10pt} \Gamma_{\rm q},\varepsilon >0,\\
\label{dads_qend}
    W_{\rm q} &= \frac{1}{2}v^2_{\rm{c,q}}+\frac{1}{2}\big(i_{\rm{t,q}} - i^{\rm r}_{\rm{t,q}}\big)^2.
\end{align}

Equations \eqref{itd_r} and \eqref{itq_r} serve as outer voltage controllers whereas  \eqref{dads_d1},\eqref{u_d}-\eqref{dads_dend} and \eqref{dads_q1},\eqref{u_q}-\eqref{dads_qend} serve as inner current controllers. The auxiliary inputs $u_{\rm d}$ and $u_{\rm q}$ inject nonlinear damping and implement the DADS mechanism. The gains $K_{\rm VC}>0$ and $K_{\rm CC}>0$ are respectively the voltage and current control gains. The parameters $\mu_{\rm d},\mu_{\rm q}>0$ tune the strength of the nonlinear damping injection and determine the trade-off between control effort and disturbance attenuation (\textit{c.f.} \eqref{re52},\eqref{re53}). The adaptive gains $z_{\rm d},z_{\rm q}$ with dynamics \eqref{z_ddot}, \eqref{z_qdot} compensate for unknown bounded grid voltage and unknown network parameters $R,L>0$. The adaptation rates $\Gamma_{\rm d},\Gamma_{\rm q}>0$ determine the speed of gain adjustment.

The DADS-BS controller \eqref{dads_d1}–\eqref{dads_qend} combines classical BS with nonlinear dynamic damping, adaptation of the gains $z_{\rm d}$ and $z_{\rm q}$, and a deadzone mechanism. The nonlinear damping terms $\frac{(1+e^{z_\ell})\omega_{\rm b}^2}{4\mu_\ell}$, $\ell\in\{\rm d,q\}$, provide disturbance suppression without requiring prior bounds on disturbance magnitude or knowledge of network parameters $R,L>0$. The deadzone in \eqref{z_ddot},\eqref{z_qdot} ensures that adaptation is activated only when the energy signals $W_{\rm d}$ and $W_{\rm q}$ are above a user-specified threshold $\varepsilon>0$, and frozen otherwise, thereby preventing gain drift and guaranteeing their boundedness.

\begin{rem} \rm For reference, a classical PI-based GFM controller
\cite{chatterjee2025effects} is given below. The controller consists of a voltage
controller that generates reference terminal currents and a current controller that
computes the terminal voltages.\\
\noindent\emph{Current controller (inner loop):}
\begin{align}
    v_{\rm t,d}
&\!\!= \!\!-K^{\rm{P}}_{\rm{CC}}\!\left(i_{\rm t,d}\!-\!i^{\rm r}_{\rm t,d}\right)
\!-\! K^{\rm{I}}_{\rm{CC}}\gamma_{\rm d}
\!+\! K^{\rm{F}}_{\rm{CC}} v_{\rm c,d}
\!-\! \omega L_{\rm f} i_{\rm t,q}, \\
v_{\rm t,q}
&\!\!=\!\!- K^{\rm{P}}_{\rm{CC}}\!\left(i_{\rm t,q}\!-\!i^{\rm r}_{\rm t,q}\right)
\!-\! K^{\rm{I}}_{\rm{CC}}\gamma_{\rm q}
\!+\! K^{\rm{F}}_{\rm{CC}} v_{\rm c,q}
\!+\! \omega L_{\rm f} i_{\rm t,d},
\end{align}
with integral states
\begin{align}
\dot{\gamma}_{\rm d} &= i_{\rm t,d} - i_{\rm t,d}^{\rm r}, \qquad
\dot{\gamma}_{\rm q} = i_{\rm t,q} - i_{\rm t,q}^{\rm r}.
\end{align}
\noindent\emph{Voltage controller (outer loop):}
\begin{align}
    i^{\rm r}_{\rm t,d}
&\!\!=\!\! -K^{\rm{P}}_{\rm{VC}}\!\left(v_{\rm c,d}\!-\!v^{\rm r}_{\rm c,d}\right)
\!-\! K^{\rm{I}}_{\rm{VC}}\beta_{\rm d}
\!+\! K^{\rm{F}}_{\rm{VC}} i_{\rm g,d}
\!-\! \omega C_{\rm f} v_{\rm c,q}, \\
i^{\rm r}_{\rm t,q}
&\!\!= \!\!-K^{\rm{P}}_{\rm{VC}}\!\left(v_{\rm c,q}\!-\!v^{\rm r}_{\rm c,q} \right)
\!-\! K^{\rm{ I}}_{\rm{VC}}\beta_{\rm q}
\!+\! K^{\rm{F}}_{\rm{VC}} i_{\rm g,q}
\!+\! \omega C_{\rm f} v_{\rm c,d},
\end{align}
with integral states
\begin{align}
\dot{\beta}_{\rm d} &= v_{\rm c,d} - v_{\rm c,d}^{\rm r}, \qquad
\dot{\beta}_{\rm q} = v_{\rm c,q}.
\end{align}
The DADS-BS preserves this voltage and current
control architecture and thus, can  be integrated into existing PI-based GFM control stacks without major modifications; see Fig. \ref{drooparch}.
\end{rem}

\subsection{Boundedness of States and Performance Guarantees}
We establish the assignable voltage regulation guarantees in this subsection, including the boundedness of states and performance guarantees.

Firstly, the DADS–BS controller ensures boundedness of all closed-loop states (\textit{c.f.} \eqref{thm1eq1},\eqref{thm1eq2}) in the presence of bounded grid disturbances, \textit{i.e.}, the closed-loop system exhibits a bounded-input bounded-state (BIBS) property for each fixed parameter vector $\boldsymbol{\Omega}\in\mathbb{R}_{>0}^6$ defined in \eqref{theta_def}. No claim of uniformity with respect to the parameters—particularly the unknown network parameters $R,L>0$—is made.

\begin{thm}
    [\it Boundedness of states]\label{thm1}  \it Consider the droop-based GFM inverter described by \eqref{inv1}-\eqref{xstte}. Fix design parameters $K_{\rm VC},K_{\rm CC},\mu_{\rm d},\mu_{\rm q},\Gamma_{\rm d},\Gamma_{\rm q},\varepsilon,\omega_{\rm pc},\omega_{\rm qc},\overline{Q}>0,\overline P \in (0,+\infty]$, and $\xi_{\rm p},\xi_{\rm q}>1$. Let the setpoints $\omega_0,V_0,P_{0},Q_{0}\in\mathbb{R}$ and the droop coefficients $K_{\rm P},K_{\rm Q}\in\mathbb{R}$ be given. Choose the control inputs $v_{\rm t,d}$ and $v_{\rm t,q}$ according to \eqref{dads_d1}-\eqref{dads_dend} and \eqref{dads_q1}-\eqref{dads_qend}, respectively. Assume the grid voltage components satisfy $v_{\rm g,d},v_{\rm g,q}\in L^\infty(\mathbb{R}_{+};\mathbb{R})$. Then, for each fixed $\boldsymbol{\Omega}\in\mathbb{R}^6_{>0}$ defined in \eqref{theta_def}, there exists a function $B_{1,\boldsymbol{\Omega}}\in C^0(\mathbb{R}^{10} \times \mathbb{R} \times \mathbb{R}\times \mathbb{R}_{+}\times \mathbb{R}_{+};\mathbb{R}_{+})$, such that, for every $\big(\boldsymbol{x}(0),z_{\rm d}(0), z_{\rm q}(0)\big) \in\mathbb{R}^{10}\times \mathbb{R}\times \mathbb{R}$, the unique absolutely continuous solution of \eqref{inv1}-\eqref{dads_qend} exists, and satisfies the following estimates for all $t\geq 0$:   
    \begin{align}\label{thm1eq1}
        \Vert \boldsymbol{x}(t)\Vert \leq B_{1,\boldsymbol{\Omega}}\big(\boldsymbol{x}(0),\!z_{\rm d}(0),\!z_{\rm q}(0),\!\Vert v_{\rm g,d}\Vert_{\infty},\!\Vert v_{\rm g,q}\Vert_{\infty}\big),
    \end{align}
and 
\begin{align}\label{thm1eq2}
\begin{split}
z_{\ell}(0)
\le z_{\ell}(t)
&\leq \lim_{s\rightarrow\infty} z_{\ell}(s) \\
&\leq B_{1,\boldsymbol \Omega}\big(\boldsymbol{x}(0),\!z_{\rm d}(0),\!z_{\rm q}(0),\!\Vert v_{\rm g,d}\Vert_{\infty}, \!\Vert v_{\rm g,q}\Vert_{\infty}\big),
\end{split}
\end{align}
where $\ell\in \{\rm d, \rm q\}$. 
\end{thm}

The proof for Theorem \ref{thm1} is provided in Section \ref{pfthm1}. 

Secondly, we prove that the DADS-BS can regulate the PCC voltage errors into arbitrarily small residual sets of size $\sqrt{2\varepsilon}$ (\textit{c.f.} \eqref{re54},\eqref{re55}), independently of grid disturbances and unknown constant network parameters $R,L>0$.  This is achieved because the adaptive gains $z_{\rm d}$ and $z_{\rm q}$ are pumped up as needed to counteract disturbances and uncertainties, \textit{i.e.}, sufficient adaptation is provided through the gain dynamics. The explicit upper bound on $z_\ell, \ell = \{\rm d,q\}$ in \eqref{thm1eq2} is $\ln\bigg(
  \max\Big\{
    \frac{\mu_{\rm \ell}\big(1+R^2+\Vert v_{\rm g,\ell}\Vert^2_{\infty}\big)}{2kL^2\varepsilon} - 1,
    e^{z_{\rm \ell}(0)}
  \Big\}
 + \frac{\Gamma_{\rm \ell}}{2k}
  \max\Big\{
    W_{\rm \ell}(0) + \frac{\mu_{\rm \ell}\big(1+R^2+\Vert v_{\rm g,\ell}\Vert^2_{\infty}\big)}{2kL^2\big(1+e^{z_{\rm \ell}(0)}\big)} - \varepsilon,
    0
  \Big\}
\bigg),$ (see the proof of Theorem \ref{thm1}). This shows that the upper bound on $z_\ell$ necessarily increases as $\varepsilon>0$ decreases, reflecting the increased adaptation required to drive the PCC voltage errors into smaller desired regions.

Furthermore, the convergence of the PCC voltage components to the regions $\vert v_{\rm{c,d}}(t) - v^{\rm r}_{\rm{c,d}}(t)\vert\leq \Big(\frac{\mu_{\rm d}\big(1+R^2+\Vert v_{\rm{g,d}}\Vert_{\infty}^2\big)}{kL^2}\Big)^{1/2}$ and $\vert v_{\rm c,q}\vert \leq \Big(\frac{\mu_{\rm q}\big(1+R^2+\Vert v_{\rm{g,q}}\Vert_{\infty}^2\big)}{kL^2}\Big)^{1/2}$ (\textit{c.f.} the decay estimates \eqref{re52},\eqref{re53}) is exponential with decay rate $k$ tunable via $K_{\rm VC},K_{\rm CC}>0$. Convergence to the ultimate residual sets $|v_{\rm c,d}-v^{\rm r}_{\rm c,d}|\leq\sqrt{2\varepsilon}$ and $|v_{\rm c,q}|\leq\sqrt{2\varepsilon}$  may be slower, but can be accelerated by increasing the adaptation gains $\Gamma_{\rm d},\Gamma_{\rm q}$ at the cost of increased control effort. 

\begin{thm}[\it Performance guarantees]\label{thm2} \it Consider the droop-based GFM inverter described by \eqref{inv1}-\eqref{xstte}. Fix design parameters $K_{\rm VC},K_{\rm CC},\mu_{\rm d},\mu_{\rm q},\Gamma_{\rm d},\Gamma_{\rm q},\varepsilon,\omega_{\rm pc},\omega_{\rm qc},\overline{Q}>0, \overline P \in (0,+\infty]$, and $\xi_{\rm p},\xi_{\rm q}>1$. Let the setpoints $\omega_0,V_0,P_{0},Q_{0}\in\mathbb{R}$ and the droop coefficients $K_{\rm P},K_{\rm Q}\in\mathbb{R}$ be given. Choose the control inputs $v_{\rm t,d}$ and $v_{\rm t,q}$ according to \eqref{dads_d1}-\eqref{dads_dend} and \eqref{dads_q1}-\eqref{dads_qend}, respectively. Assume the grid voltage components satisfy $v_{\rm g,d},v_{\rm g,q}\in L^\infty(\mathbb{R}_{+};\mathbb{R})$ and that the network parameters
$R,L>0$ are unknown but constant. For any initial condition $\big(\boldsymbol{x}(0),z_{\rm d}(0),z_{\rm q}(0) \big) \in\mathbb{R}^{10}\times \mathbb{R}\times\mathbb{R}$, let
$t\mapsto\big(\boldsymbol{x}(t),z_{\rm d}(t),z_{\rm q}(t)\big)$ denote the corresponding unique absolutely continuous closed-loop solution on $[0,\infty)$. Then, the following estimates hold for all $t\geq 0$:
\begin{align}
     \label{re52}
\begin{split}
&\big( v_{\rm{c,d}}(t) - v^{\rm r}_{\rm{c,d}}(t)\big)^2\\&\quad \leq 2W_{\rm d}(t)\le  2e^{-2kt}W_{\rm d}(0)+ \frac{\mu_{\rm d}\big(1+R^2+\Vert v_{\rm{g,d}}\Vert_{\infty}^2\big)}{kL^2},
\end{split}\\
    \label{re53}
\begin{split}
& v^2_{\rm{c,q}}(t) \leq 2W_{\rm q}(t)\le  2e^{-2kt}W_{\rm q}(0)+ \frac{\mu_{\rm q}\big(1+R^2+\Vert v_{\rm{g,q}}\Vert_{\infty}^2\big)}{kL^2},
\end{split}
\end{align}
where $k>0$ is 
\begin{align}\label{decay_k}
    k = \min\{K_{\rm VC},K_{\rm CC}\}. 
\end{align}
Moreover,
\begin{align}
    \label{re54}
\limsup_{t\to+\infty}\big(\vert v_{\rm{c,d}}(t) - v^{\rm r}_{\rm{c,d}}(t)\vert\big)&\le \sqrt{2\varepsilon },\\
\label{re55}
\limsup_{t\to+\infty}\big(\vert v_{\rm{c,q}}(t)\vert\big)&\le \sqrt{2\varepsilon }.
\end{align}
\end{thm}

The proof for Theorem \ref{thm2} is provided in Section \ref{pfthm2}.

\begin{rem}\rm 
The performance guarantees \eqref{re52}-\eqref{re55} remain valid even under severe fault conditions, such as three-phase-to-ground faults at the grid, \textit{i.e.},
$\sqrt{v_{\rm g,d}^2(t) + v_{\rm g,q}^2(t)} \equiv 0$. Indeed, such faulted grid voltages still belong to $L^\infty(\mathbb{R}_{+};\mathbb{R})$, which is the only assumption imposed on the grid voltage signals in the analysis.
\end{rem}

\section{Over-current Limiting with Control Barrier Functions}
\label{secCBF}

In this section, we address the problem of enforcing hard current limits on the inverter terminal current. To this end, we incorporate a CBF based safety filter that admits an arbitrary locally Lipschitz nominal controller and guarantees that the magnitude of the inverter terminal current, $i_{\rm t} = i_{\rm t,d} + j i_{\rm t,q}$, does not exceed a prescribed bound \( I_{\max} \). The proposed construction is independent of the specific nominal controller and therefore applies, in particular, to the proposed DADS-BS controller. Figure \ref{sftyfigg} illustrates the CBF-based modification to the current controller module in Fig. \ref{drooparch}, ensuring that the terminal-current magnitude remains below $I_{\rm max}$. We first briefly review the CBF framework, then specialize it to the terminal-current dynamics of the GFM inverter, and finally establish forward invariance of the safe-current set together with boundedness of the remaining closed-loop states. We also address the issue of boundedness of the adaptive gains $(z_{\rm d}, z_{\rm q})$ when the nominal controller is chosen as DADS-BS.

\subsection{Safety-critical control with locally Lipschitz nominal controllers}
\begin{figure}
    \centering
    \includegraphics[width=1\linewidth]{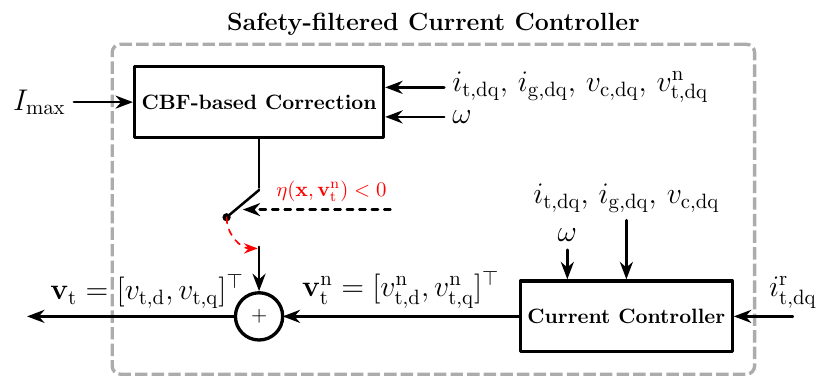}
    \caption{CBF-based safety filtering of the current controller. 
The CBF correction is activated when $\eta(\boldsymbol{x},\boldsymbol{v}_{\rm t}^{\rm n})<0$; see \eqref{sfty_ip1}.}
    \label{sftyfigg}
\end{figure}
\begin{df}[\it Extended class $\mathcal{K}$ functions]\rm  A continuous function $\alpha : [-a_0, a_1) \to [-\infty, \infty)$, $a_0 > 0, a_1 > 0$ belongs to extended class $\mathcal{K}$ if it is strictly increasing and $\alpha(0) = 0$.
\end{df}
We consider a general nonlinear control-affine system of the form
\begin{align}\label{ysys}
    \dot{\boldsymbol y} = \boldsymbol f(\boldsymbol y) + \boldsymbol g(\boldsymbol y)\boldsymbol{u},
\end{align}
where $\boldsymbol f : \mathbb R^n \to \mathbb R^n$ and $\boldsymbol g : \mathbb R^n \to \mathbb R^{n \times m}$ are smooth maps, $\boldsymbol y \in \mathbb R^n$ is the state, and $\boldsymbol u \in \mathbb R^m$ is the control input.

\begin{df}[\it Forward invariant set]
A set $\mathcal C_{\boldsymbol{y}} \subseteq \mathbb{R}^n$ is forward invariant for system \eqref{ysys} if every solution with initial condition $\boldsymbol y(0) \in \mathcal C_{\boldsymbol y}$ satisfies $\boldsymbol y(t) \in \mathcal C_{\boldsymbol y}$ for all $t \ge 0$.
\end{df}

\begin{df}[\it Control barrier function \cite{ames2016control}] Let
\begin{align*}
    \mathcal{C}_{\boldsymbol{y}} := \{\boldsymbol y \in \mathbb R^n : h(\boldsymbol y)\ge 0\},
\end{align*}
for a continuously differentiable function $h : \mathbb R^n \to \mathbb R$. The function $h$ is a control barrier function (CBF) for system \eqref{ysys} on $\mathcal C_{\boldsymbol y}$ if there exists an extended class $\mathcal{K}$ function $\alpha$ such that
\begin{align*}
\sup_{\boldsymbol u} \Bigl[ L_{\boldsymbol f} h(\boldsymbol y) + L_{\boldsymbol g} h(\boldsymbol y)\,\boldsymbol u + \alpha\bigl(h(\boldsymbol y)\bigr) \Bigr] \geq 0,
\quad \forall \boldsymbol y \in \mathcal C_{\boldsymbol y},
\end{align*}
and $L_{\boldsymbol g} h(\boldsymbol y) \neq 0$ whenever $h(\boldsymbol y)=0$. Here, $L_{\boldsymbol f}$ and $L_{\boldsymbol g}$ denote the Lie derivatives along $\boldsymbol f$ and $\boldsymbol g$, respectively.
\end{df}

\begin{thm}[\it \cite{ames2016control}]
Suppose $h$ is a CBF for \eqref{ysys} on $\mathcal C_{\boldsymbol y}$. Let $\alpha$ be the associated extended class $\mathcal K$ function. If a locally Lipschitz feedback law $\boldsymbol u =  \boldsymbol k (\boldsymbol y)$ with $\boldsymbol k : \mathbb{R}^n \to \mathbb{R}^m$  satisfies
\begin{align*}
    \dot{h}(\boldsymbol y, \boldsymbol u) \geq -\alpha\bigl(h(\boldsymbol y)\bigr),
    \quad \forall \boldsymbol y \in \mathcal C_{\boldsymbol y},
\end{align*}
then the set $\mathcal C_{\boldsymbol y}$ is forward invariant for the closed-loop system.
\end{thm}

We now apply the above framework to the terminal-current dynamics of the droop-based GFM inverter. Concatenating the current dynamics \eqref{itddn},\eqref{itqdn}, we obtain
\begin{align*}
    \dot{\boldsymbol i}_{\rm t} &= \boldsymbol A(\boldsymbol x) \boldsymbol i_{\rm t}-\frac{\omega_{\rm b}}{L_{\rm f}}\boldsymbol v_{\rm c}+\frac{\omega_{\rm b}}{L_{\rm {\rm f}}}\boldsymbol v_{\rm t},
\end{align*}
where 
\begin{align*}
    \boldsymbol{i}_{\rm t} := \begin{bmatrix}
i_{\rm{t,d}}\\ i_{\rm{t,q}} 
\end{bmatrix},\hspace{5pt} 
\boldsymbol{v}_{\rm c}:=\begin{bmatrix}
    v_{\rm{c,d}}\\v_{\rm{c,q}}
\end{bmatrix}, \hspace{5pt}  \hspace{5pt} \boldsymbol v_{\rm t}:=\begin{bmatrix}
    v_{\rm{t,d}}\\v_{\rm{t,q}}
\end{bmatrix},
\end{align*}
and 
\begin{align*}
\boldsymbol A(\boldsymbol x) = \begin{bmatrix}
 - \frac{R_{\rm f}}{L_{\rm f}} \omega_{ \rm b}  && \omega_{ \rm b} \omega\\
  -\omega_{ \rm b} \omega && - \frac{R_{\rm f}}{L_{\rm f}} \omega_{\rm b}
\end{bmatrix}.
\end{align*}
Here, $\boldsymbol{x}$ denotes the state vector given by \eqref{xstte}. 

We enforce the constraint that the magnitude of the terminal current is bounded by $I_{\max}>0$, that is,
\begin{align*}
    \|\boldsymbol i_{\rm t}\|^2 = i_{\rm t,d}^2 + i_{\rm t,q}^2 \le I_{\max}^2.
\end{align*}
This defines the safe set
\begin{align}\label{sfstt}
    \mathcal C := \bigl\{\boldsymbol x \in \mathbb R^{10} | i_{\rm t,d}^2 + i_{\rm t,q}^2 \le I_{\max}^2\bigr\}.
\end{align}
To encode this requirement as a CBF condition, we introduce the barrier function
\begin{align}\label{brfntn}
    h(\boldsymbol x) := I_{\max}^2 - i_{\rm t,d}^2 - i_{\rm t,q}^2 =  I_{\max}^2 - \boldsymbol i_{\rm t}^\top \boldsymbol i_{\rm t}.
\end{align}
Differentiating $h$ along the trajectories of the system yields
\begin{align*}
\begin{split}
    \dot{h}(\boldsymbol x) &= -2 \boldsymbol i_{\rm t}^\top\dot{\boldsymbol i}_{\rm t}\\
    &=-2\boldsymbol i_{\rm t}^\top \boldsymbol A(\boldsymbol x) \boldsymbol i_{\rm t} + \frac{2\omega_{\rm b}}{L_{\rm f}}\boldsymbol i^\top_{\rm t} \boldsymbol v_{\rm c} - \frac{2\omega_{\rm b}}{L_{\rm f}}\boldsymbol i^\top_{\rm t} \boldsymbol v_{\rm t}.
\end{split}
\end{align*}
Since $\dot h$ depends affinely on $\boldsymbol v_{\rm t}$, we can enforce a CBF inequality of the form
\begin{align}\label{cbf-ineq-h}
\dot h(\boldsymbol x) \geq -ch(\boldsymbol x),
\end{align}
for some design constant $c>0$, by constraining $\boldsymbol v_{\rm t}$.

Let the nominal control be denoted by 
\begin{align*}
    \boldsymbol v^{\rm n}_{\rm t} = \begin{bmatrix}
        v^{\rm n}_{\rm t,d}\\v^{\rm n}_{\rm t,q}
    \end{bmatrix}.
\end{align*}
This nominal control may be any locally Lipschitz function of the state.  We then construct a safety filter that modifies $\boldsymbol v_{\rm t}^{\rm n}$ only as needed to satisfy \eqref{cbf-ineq-h}. To this end, we solve the quadratic program (QP)

\begin{align}\label{qpcnstnt}
    \begin{aligned}
        \min_{\boldsymbol v_t} \quad & \|\boldsymbol v_{\rm t} - \boldsymbol v^{\rm n}_{\rm t}\|^2 \\
        \text{s.t.} \quad & -2\boldsymbol i_{\rm t}^\top \boldsymbol A(\boldsymbol x) \boldsymbol i_{\rm t} + \frac{2\omega_{\rm b}}{L_{\rm f}}\boldsymbol i^\top_{\rm t} \boldsymbol v_{\rm c} - \frac{2\omega_{\rm b}}{L_{\rm f}}\boldsymbol i^\top_{\rm t} \boldsymbol v_{\rm t} \geq - c h(\boldsymbol x),
    \end{aligned}
\end{align}
which explicitly enforces the CBF condition with $\alpha(s)=c s$. The problem \eqref{qpcnstnt} is strictly convex with a single affine inequality constraint and thus admits a unique optimal solution.

Define
\begin{align}\label{sfty_ip2}
\eta(\boldsymbol x, \boldsymbol v^{\rm n}_{\rm t}) :=  -2\boldsymbol i_{\rm t}^\top \boldsymbol A(\boldsymbol x) \boldsymbol i_{\rm t} + \frac{2\omega_{\rm b}}{L_{\rm f}}\boldsymbol i^\top_{\rm t} \boldsymbol v_{\rm c} - \frac{2\omega_{\rm b}}{L_{\rm f}}\boldsymbol i^\top_{\rm t} \boldsymbol v_{\rm t}^{\rm n}  + c h(\boldsymbol x).
\end{align}
The nominal control $\boldsymbol v_{\rm t}^{\rm n}$ satisfies the constraint in \eqref{qpcnstnt} if and only if $\eta(\boldsymbol x,\boldsymbol v_{\rm t}^{\rm n})\ge 0$. Otherwise, the constraint is violated, in which case, the optimal solution of \eqref{qpcnstnt} is given by the projection of $\boldsymbol v_{\rm t}^{\rm n}$ onto the half-space defined by the affine constraint.  This results in the explicit solution
\begin{align}\label{sfty_ip1}
\boldsymbol v_{\rm t} = 
\begin{cases}
\boldsymbol v^{\rm n}_{\rm t}, & \eta(\boldsymbol x, \boldsymbol v^{\rm n}_{\rm t}) \geq 0, \\[6pt]
\boldsymbol v^{\rm n}_{\rm t} + \dfrac{L_{\rm f}}{2\omega_{\rm b}} \dfrac{\eta(\boldsymbol x, \boldsymbol v^{\rm n}_{\rm t})}{\|  \boldsymbol i_{\rm t} \|^2}\boldsymbol i_{\rm t}, & \text{otherwise}.
\end{cases}
\end{align}
See Fig. \ref{sftyfigg} also. When $\|\boldsymbol i_{\rm t}\| = 0$, the CBF constraint in \eqref{qpcnstnt} reduces to $0 \ge -c I_{\max}^2$, which is trivially satisfied since $I_{\max}^2 > 0$. Thus, when 
$\boldsymbol i_{\rm t} = 0$, the CBF condition imposes no restriction on 
$\boldsymbol v_{\rm t}$, and therefore the nominal control 
$\boldsymbol v_{\rm t}^{\rm n}$ is applied without modification.

Below, for the ease of stating results, we collect the states excluding terminal current components $(i_{\rm t,d},i_{\rm t,q})$ into $\boldsymbol{x}'$ as 
\begin{align}\label{xbar}
    \boldsymbol x' := \big[v_{\rm {c,d}}, \, v_{\rm {c,q}}, \,  i_{\rm {g,d}}, \,  i_{\rm {g,q}},\, q_1 ,\,  q_2, \, p_1,\, p_2\big]^{  \top} \in \mathbb{R}^{8}. 
\end{align}

\begin{thm}
 [\it Safety with locally Lipschitz nominal controllers] \label{thm4}   Consider the droop-based GFM inverter described by \eqref{inv1}-\eqref{xstte}. Fix design parameters $\omega_{\rm pc},\omega_{\rm qc}>0$, $\xi_{\rm p},\xi_{\rm q}>1$, and $\overline P, \overline Q \in (0,+\infty]$. Let the setpoints $\omega_0,V_0,P_{0},Q_{0}\in\mathbb{R}$, the droop coefficients $K_{\rm P},K_{\rm Q}\in\mathbb{R},$ and the maximum allowable terminal current $I_{\max}>0$ be given. Let the current-limiting barrier function  $h(\boldsymbol x): \mathbb{R}^{10}\rightarrow \mathbb{R}$ be given by \eqref{brfntn}. Suppose the terminal-voltage control $\boldsymbol{v}_{\rm t} = [v_{\rm t,d},v_{\rm t,q}]^\top$ is implemented according to \eqref{sfty_ip2},\eqref{sfty_ip1}, for an arbitrary locally Lipschitz nominal control
$\boldsymbol v_{\rm t}^{\rm n}=\boldsymbol k(\boldsymbol x)$,
with $\boldsymbol k:\mathbb{R}^{10}\rightarrow\mathbb{R}^2$. Assume the grid voltage components satisfy $v_{\rm g,d},v_{\rm g,q}\in L^\infty(\mathbb{R}_{+};\mathbb{R})$. Then, for every initial condition $\boldsymbol{x}(0) \in\mathcal{C}$, with $\mathcal{C}$ defined in \eqref{sfstt}, the closed-loop system admits a unique absolutely continuous solution for all $t\geq 0$, and  
\begin{align}\label{crrntsft}
    \sqrt{i^2_{\rm t,d}(t)+i^2_{\rm t,q}(t)}\leq I_{\rm max},
\end{align}
all $t\geq 0$. 
Moreover, for each fixed $\boldsymbol{\Omega}\in\mathbb{R}^6_{>0}$ defined in \eqref{theta_def}, there exists a function $B_{2,\boldsymbol{\Omega}}\in C^0(\mathbb{R}^{10}\times \mathbb{R}_{+}\times \mathbb{R}_{+})$, such that
    \begin{align}\label{xdshbnd}
        \Vert \boldsymbol{x}'(t)\Vert \leq B_{2,\boldsymbol{\Omega}}\big(\boldsymbol{x}(0),\Vert v_{\rm g,d}\Vert_{\infty},\Vert v_{\rm g,q}\Vert_{\infty}\big),
    \end{align}
for all $t\geq 0$, where $\boldsymbol x'$ is given by \eqref{xbar}. 
\end{thm}

The proof for Theorem \ref{thm4} is provided in Section \ref{pfthm4}. 

\subsection{Safe DADS-BS GFM control}
We now specialize to the case where the nominal controller $\boldsymbol v_{\rm t}^{\rm n}$ is chosen as the DADS-BS.

\begin{prop}\label{propDADS}
     Consider the droop-based GFM inverter described by \eqref{inv1}-\eqref{xstte}. Fix design parameters $K_{\rm VC},K_{\rm CC},\mu_{\rm d},\mu_{\rm q},\Gamma_{\rm d},\Gamma_{\rm q},\varepsilon,\omega_{\rm pc},\omega_{\rm qc}>0$, $\xi_{\rm p},\xi_{\rm q}>1$, and and $\overline P, \overline Q \in (0,+\infty]$. Let the setpoints $\omega_0,V_0,P_{0},Q_{0}\in\mathbb{R}$, the droop coefficients $K_{\rm P},K_{\rm Q}\in\mathbb{R},$ and the maximum allowable terminal current $I_{\max}>0$ be given. Let the current-limiting barrier function  $h(\boldsymbol x): \mathbb{R}^{10}\rightarrow \mathbb{R}$ be given by \eqref{brfntn}. Suppose the terminal-voltage control $\boldsymbol{v}_{\rm t}=[v_{\rm t,d},v_{\rm t,q}]^\top$ is implemented according to \eqref{sfty_ip2},\eqref{sfty_ip1}, where the nominal control $\boldsymbol{v}_{\rm t}^{\rm n}=[v_{\rm t,d}^{\rm n}, v_{\rm t,q}^{\rm n}]^\top$ is the DADS-BS defined in \eqref{dads_d1}-\eqref{dads_qend}. Assume the grid voltage components satisfy $v_{\rm g,d},v_{\rm g,q}\in L^\infty(\mathbb{R}_{+};\mathbb{R})$. Then, for each fixed $\boldsymbol{\Omega}\in\mathbb{R}^6_{>0}$ defined in \eqref{theta_def} and for every initial condition $\boldsymbol{x}(0) \in\mathcal{C}$ (with $\mathcal C$ defined in \eqref{sfstt})  and $z_{\rm d}(0),z_{\rm q}(0) \in\mathbb{R}$, the closed-loop system admits a unique  absolutely continuous solution for all $t\geq 0$.
\end{prop}

The proof for Proposition \ref{propDADS} is provided in Section \ref{pfprp1}.

To show explicitly that the adaptive gains $z_{\rm d}$ and $z_{\rm q}$ are bounded under the safety-filtered DADS-BS control, we analyze the activation pattern of the safety filter.

For a given solution $t \mapsto (\boldsymbol{x}(t),z_{\rm d}(t),z_{\rm q}(t))$ of the safety-filtered DADS-BS closed-loop system
\eqref{inv1}-\eqref{dads_qend},\eqref{sfty_ip2},\eqref{sfty_ip1} define
\begin{align}
\eta(t):=\eta\big(\boldsymbol x(t),\boldsymbol{v}_{\rm t}^{\rm n}(t)\big).
\end{align}
Since $\boldsymbol x(\cdot)$ and $\boldsymbol v_{\rm t}^{\rm n}(\cdot)$ are continuous and the map
$(\boldsymbol x,\boldsymbol v_{\rm t}^{\rm n})\mapsto \eta(\boldsymbol x,\boldsymbol v_{\rm t}^{\rm n})$ is locally Lipschitz,
it follows that $\eta(\cdot)$ is continuous on $\mathbb{R}_{+}$.

We define the ON set
\begin{align*}
    \mathcal I_{\rm on}:= \{ t \ge 0 \vert \eta(t) < 0 \},
\end{align*}
that is, the set of times at which the safety filter is active and the
implemented control differs from the nominal one. If $\eta(0)<0$, then there exists $\sigma_0>0$ such that
$[0,\sigma_0)\subset \mathcal I_{\rm on}$.  
After this possible initial ON interval, all remaining ON episodes are ordinary
open intervals, because $\eta(\cdot)$ is continuous.  Thus $\mathcal I_{\rm on}$ can be decomposed as
\begin{align}\label{act_intvl}
    \mathcal I_{\rm on}
    = \bigl([0,\sigma_0) \;\text{if }\eta(0)<0\bigr)
      \;\cup\;\bigcup_{\kappa\in\mathcal J}(\tau_\kappa,\sigma_\kappa),
\end{align}
where $\mathcal J \subseteq \mathbb N$ (possibly empty) and each
interval $(\tau_\kappa,\sigma_\kappa)$ satisfies $0\!<\!\tau_\kappa\!<\!\sigma_\kappa\!\leq\!+\infty$.
Each interval in \eqref{act_intvl} is interpreted as one ON episode of the
safety filter: at $t=\tau_\kappa$ the filter turns ON (\textit{i.e.}, $\eta$ becomes
negative), for $t\in(\tau_\kappa,\sigma_\kappa)$ it remains ON, and at $t=\sigma_\kappa$ it
turns OFF when $\eta(t)$ returns to $0$.

We define the number of ON episodes and the total ON time by
\begin{align}
    N_{\eta} 
    &:= 
    \begin{cases}
        1 + \mathrm{card}(\mathcal{J}), & \text{if }\eta(0)<0,\\
        \mathrm{card}(\mathcal{J}),     & \text{if }\eta(0)\ge 0,
    \end{cases}
    \label{Neta_def} \\
    T_{\eta} 
    &:=
    \begin{cases}
        \sigma_0 + \displaystyle\sum_{\kappa\in\mathcal{J}}(\sigma_\kappa-\tau_\kappa), 
            & \text{if }\eta(0)<0,\\
        \displaystyle\sum_{\kappa\in\mathcal{J}}(\sigma_\kappa-\tau_\kappa),
            & \text{if }\eta(0)\ge 0,
    \end{cases}
    \label{Teta_def}
\end{align}
so that $N_\eta$ counts all ON episodes and $T_\eta$ gives their total length. Here, $\rm card(\cdot)$ denotes the cardinality of a set. 

\begin{asm}\label{asm:eta_events}
For every initial condition $\boldsymbol x(0)\in\mathcal C$ and $z_{\rm d}(0),z_{\rm q}(0)\in\mathbb R$,
and for every bounded grid voltage $v_{\rm g,d},v_{\rm g,q}\in L^\infty(\mathbb R_+;\mathbb R)$,
consider the corresponding closed-loop solution of \eqref{inv1}-\eqref{dads_qend},\eqref{sfty_ip1},\eqref{sfty_ip2} and define $N_\eta$ and $T_\eta$ by
\eqref{Neta_def},\eqref{Teta_def}.
Then, for each fixed $\boldsymbol{\Omega}\in\mathbb{R}^6_{>0}$ defined in \eqref{theta_def}, there exist $N_{\boldsymbol{\Omega}}\in\mathbb N$ and $T_{\boldsymbol{\Omega}}\in(0,\infty)$ possibly depending on $\boldsymbol{x}(0),z_{\rm d}(0),z_{\rm q}(0),\Vert v_{\rm g,d}\Vert_{\infty},\Vert v_{\rm g,q}\Vert_{\infty}$, such that 
\begin{align*}
N_\eta \le N_{\boldsymbol \Omega},\qquad T_\eta \le T_{\boldsymbol \Omega}.
\end{align*}
\end{asm}

Assumption \ref{asm:eta_events} states that, along any closed-loop trajectory
with $\boldsymbol x(0)\in\mathcal C,z_{\rm d}(0),z_{\rm q}(0) \in\mathbb{R}$ and $v_{\rm g,d},v_{\rm g,q}\in L^{\infty}(\mathbb{R_+};\mathbb{R})$, the safety filter is activated only finitely many times and that the total duration of safety–filter activation is finite.

\begin{rem}\rm \label{rmksfz}
Assumption \ref{asm:eta_events} imposes two constraints on the set $\mathcal I_{\rm on}$ of safety–filter ON times.
First, $N_\eta\leq N_{\boldsymbol \Omega}<\infty$ implies that $\mathcal I_{\rm on}$ consists of at most $N_{\boldsymbol \Omega}$
disjoint ON intervals: possibly an initial interval $[0,\sigma_0)$ if
$\eta(0)<0$, followed by the intervals $(\tau_\kappa,\sigma_\kappa)$,
$\kappa\in\mathcal J$. Since $T_\eta\le T_{\boldsymbol \Omega}<\infty$, each of these intervals has
finite length, and therefore every right–endpoint $\sigma_0$ (if present) and
$\sigma_\kappa$ satisfies $
\sigma_0<+\infty, \qquad \sigma_\kappa<+\infty, \quad \kappa\in\mathcal J.
$

Hence, the set of all right–endpoints,
\begin{align*}
\Sigma:=\{\sigma_0 \;\text{if }\eta(0)<0\}\cup\{\sigma_\kappa \mid \kappa\in\mathcal J\},
\end{align*}
is finite and possibly empty, and we define
\begin{align*}
\sigma_{\max}:=\sup(\Sigma),
\end{align*}
with the convention $\sup\varnothing := 0$, so that $\sigma_{\max}<+\infty$. For all $t\ge\sigma_{\max}$ we therefore have $t\notin\mathcal I_{\rm on}$, or equivalently,
\begin{align*}
t\ge \sigma_{\max} \quad\Longrightarrow\quad \eta(t)\ge 0.
\end{align*}
Thus, under Assumption \ref{asm:eta_events}, the safety filter can be activated
only finitely many times, and there exists a finite time after which it remains
permanently OFF. In particular, if the safety filter is never activated (\textit{i.e.}, $\mathcal I_{\rm on}=\varnothing$), then $\sigma_{\max}=0$.
\end{rem}

Below we establish the main results under the safety-filtered DADS-BS control. 

\begin{thm}[\it Safety with DADS-BS nominal controller]\label{thm6}
    Consider the droop-based GFM inverter described by \eqref{inv1}-\eqref{xstte}. Fix design parameters $K_{\rm VC},K_{\rm CC},\mu_{\rm d},\mu_{\rm q},\Gamma_{\rm d},\Gamma_{\rm q},\varepsilon,\omega_{\rm pc},\omega_{\rm qc}>0$, $\xi_{\rm p},\xi_{\rm q}>1$, and $\overline P, \overline Q \in (0,+\infty]$. Let the setpoints $\omega_0,V_0,P_{0},Q_{0}\in\mathbb{R}$, the droop coefficients $K_{\rm P},K_{\rm Q}\in\mathbb{R},$ and the maximum allowable terminal current $I_{\max}>0$ be given. Let the current-limiting barrier function  $h(\boldsymbol x): \mathbb{R}^{10}\rightarrow \mathbb{R}$ be given by \eqref{brfntn}. Suppose the terminal-voltage control $\boldsymbol{v}_{\rm t}=[v_{\rm t,d},v_{\rm t,q}]^\top$ is implemented according to \eqref{sfty_ip2},\eqref{sfty_ip1} where the nominal control $\boldsymbol{v}_{\rm t}^{\rm n}=[v_{\rm t,d}^{\rm n}, v_{\rm t,q}^{\rm n}]^\top$ is the DADS-BS defined in \eqref{dads_d1}-\eqref{dads_qend}. Assume the grid voltage components satisfy $v_{\rm g,d},v_{\rm g,q}\in L^\infty(\mathbb{R}_{+};\mathbb{R})$. For each fixed $\boldsymbol{\Omega}\in\mathbb{R}^6_{>0}$ defined in \eqref{theta_def} and any initial condition $\boldsymbol{x}(0) \in\mathcal{C}$ (with $\mathcal C$ defined in \eqref{sfstt})  and $z_{\rm d}(0),z_{\rm q}(0) \in\mathbb{R}$, let
$t\mapsto\big(\boldsymbol{x}(t),z_{\rm d}(t),z_{\rm q}(t)\big)$ denote the unique absolutely continuous closed-loop solution on $[0,\infty)$. Then the solution satisfies \eqref{crrntsft} and \eqref{xdshbnd} for all $t\geq 0$. Moreover, subject to Assumption \ref{asm:eta_events}, there exists a function   $B_{3,\boldsymbol \Omega}: \mathbb{R}^{10}\times \mathbb{R}\times \mathbb{R} \times \mathbb{R}_{+}\times \mathbb{R}_{+}\times \mathbb{N}\times \mathbb{R}_+\to\mathbb{R}_{+}$, such that for each $\ell\in \{\rm d, \rm q\}$, the adaptive gains satisfy
    \begin{align}\label{zqdbndsf}
\begin{split}
z_{\ell}(0)
\!\le \!z_{\ell}(t)
\!\leq\! \lim_{s\rightarrow\infty} z_{\ell}(s) 
\leq B_{3,\boldsymbol \Omega}\big(&\boldsymbol{x}(0),\!z_{\rm d}(0),\! z_{\rm q}(0),\Vert v_{\rm g,d}\Vert_{\infty},\\&\quad\Vert v_{\rm g,q}\Vert_{\infty},\!N_{\boldsymbol{\Omega}},\!T_{\boldsymbol{\Omega}}\big),
\end{split}
\end{align}
for all $t\geq 0$. 
\end{thm}
The proof for Theorem \ref{thm6} is provided in Section \ref{pfthm5}.

\begin{rem}\rm
Unlike in Theorems \ref{thm1} and \ref{thm2}, in Theorems \ref{thm4} and \ref{thm6}  we allow both $\overline P,\overline Q\in(0,+\infty]$. This is because the CBF--QP safety filter enforces the hard constraint $\|\boldsymbol i_{\rm t}\|\le I_{\max}$ for all time, which prevents terminal-current blow-up regardless of whether the power-filter inputs are saturated. As a consequence, the remaining electrical states are bounded under the safety-filtered control, and saturation of the reactive- and active-power inputs in \eqref{Q_filter},\eqref{P_filter} is not required for the boundedness conclusions of Theorems \ref{thm4} and \ref{thm6}.
\end{rem}

\section{Numerical Results}
\label{secSim}

\begin{figure}
    \centering
    \includegraphics[width=1\linewidth]{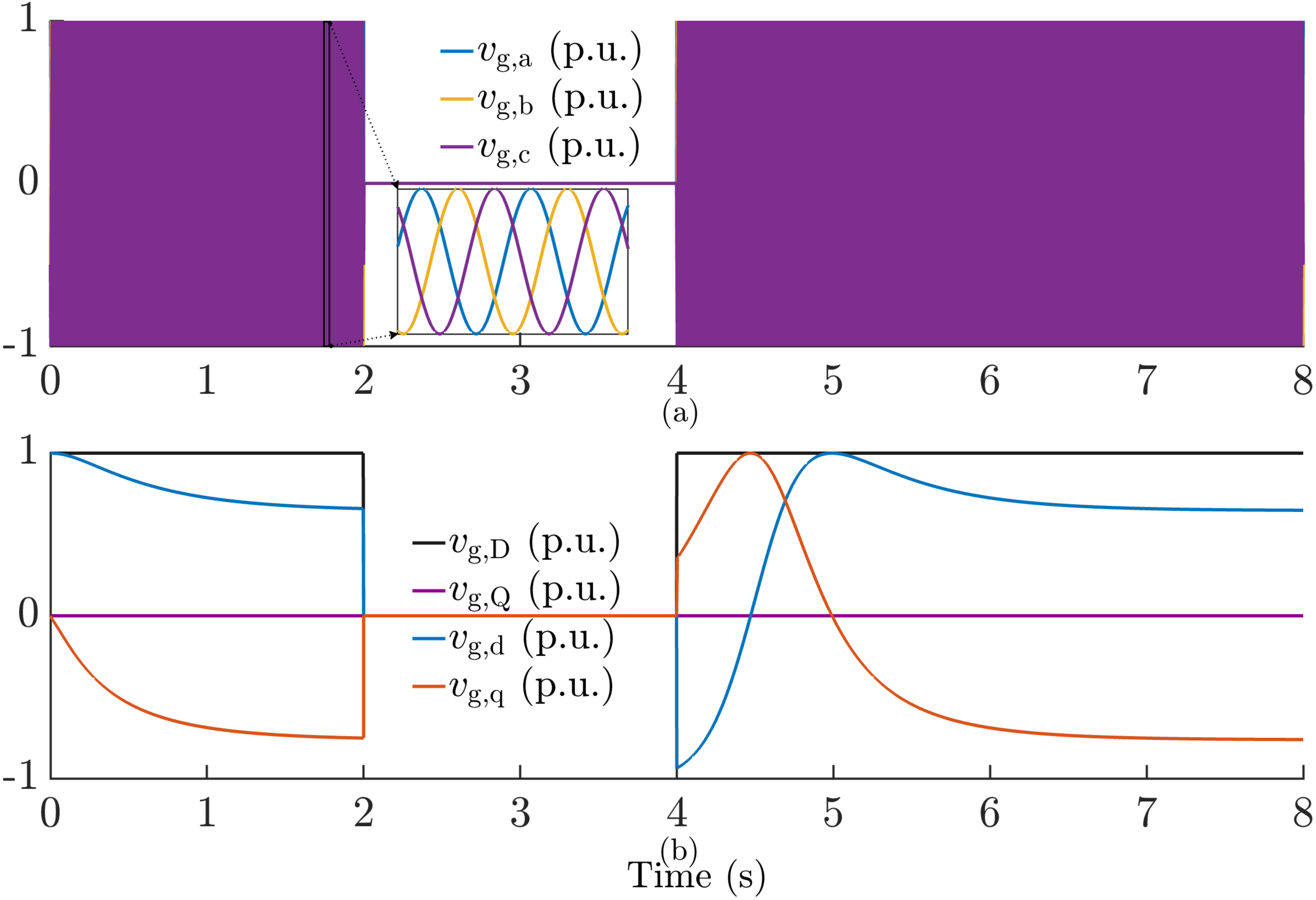}
 \caption{Grid voltage. A three-phase-to-ground fault occurs from $t = 2\,\rm{s}$ to $t = 4\,\rm{s}$. (a) Three-phase voltages. (b) Global $\rm{DQ}$ and local $\rm{dq}$  components.}
    \label{grdvol}
\end{figure}

\begin{figure}
    \centering
    \includegraphics[width=1\linewidth]{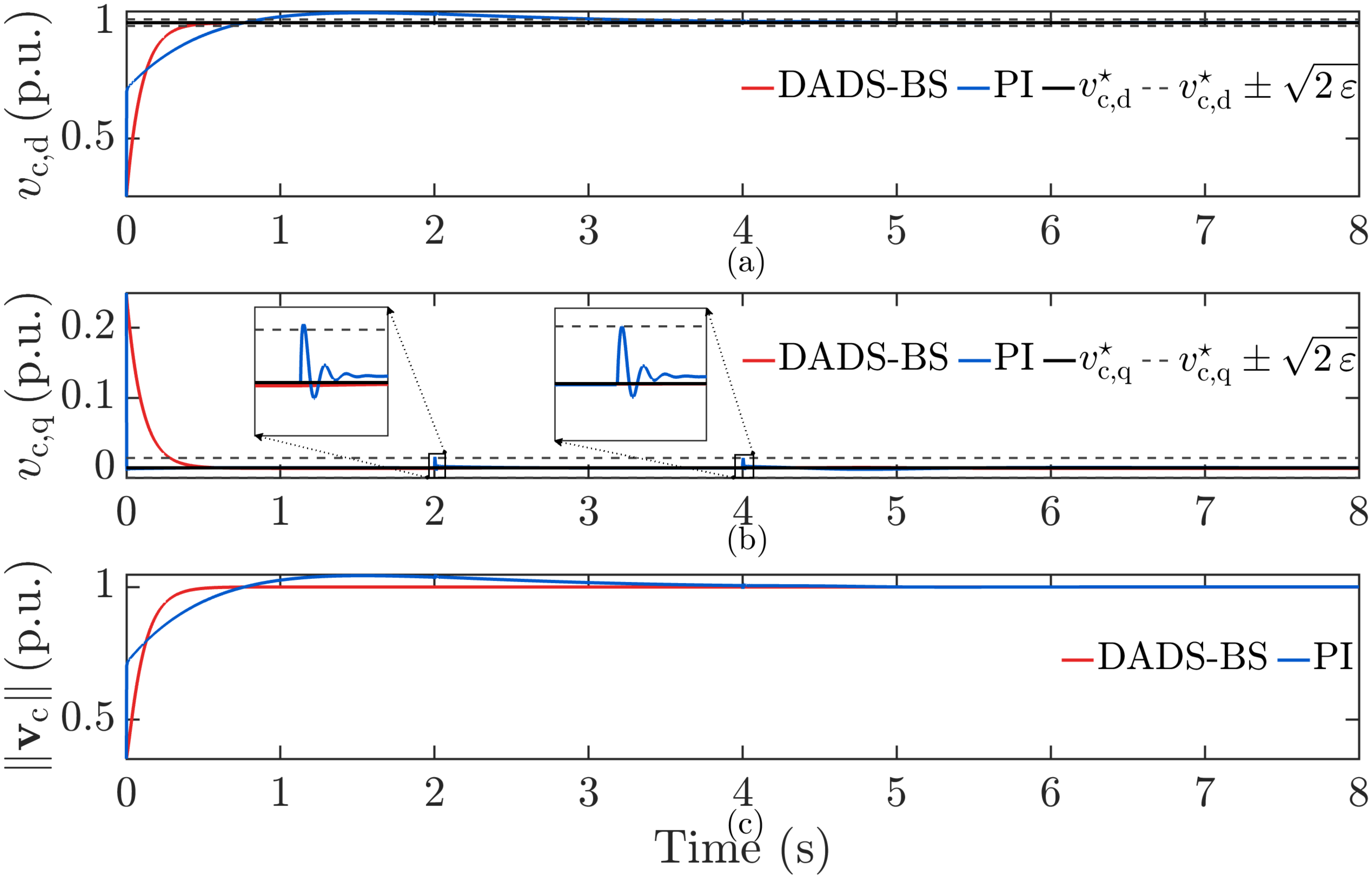}
\caption{PCC voltage behavior \textbf{without} current limiting. (a) $\rm d$ component. (b) $\rm q$ component. (c) $\boldsymbol{v}_{\rm c} = \sqrt{v_{\rm c,d}^2+v_{\rm c,q}^2}$.}
    \label{vcdq_nmnl}
\end{figure}

\begin{figure}
    \centering
    \includegraphics[width=1\linewidth]{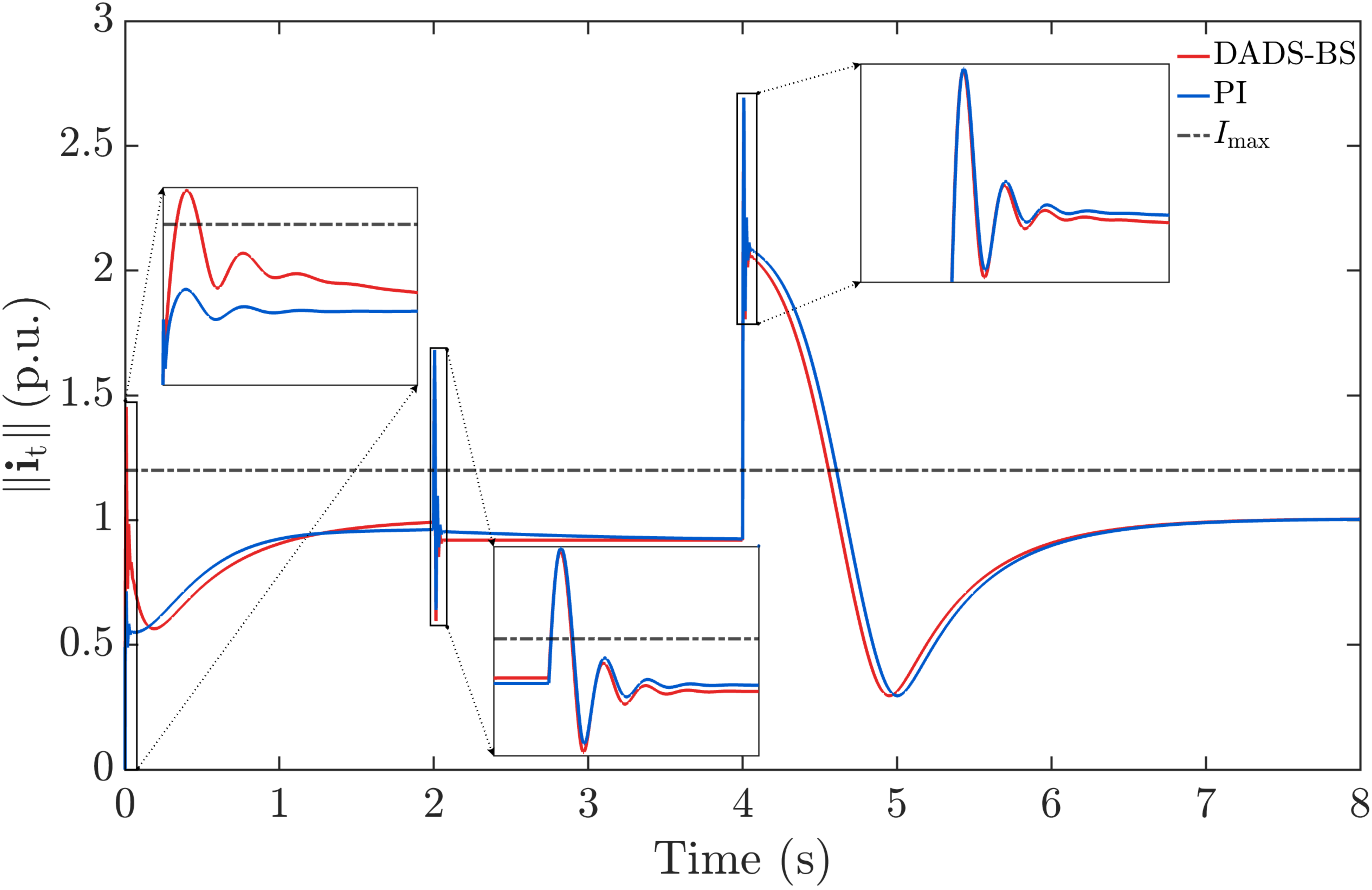}
\caption{Terminal current behavior \textbf{without} current limiting.}
    \label{termcur_nmnl}
\end{figure}

\begin{figure}
    \centering
    \includegraphics[width=1\linewidth]{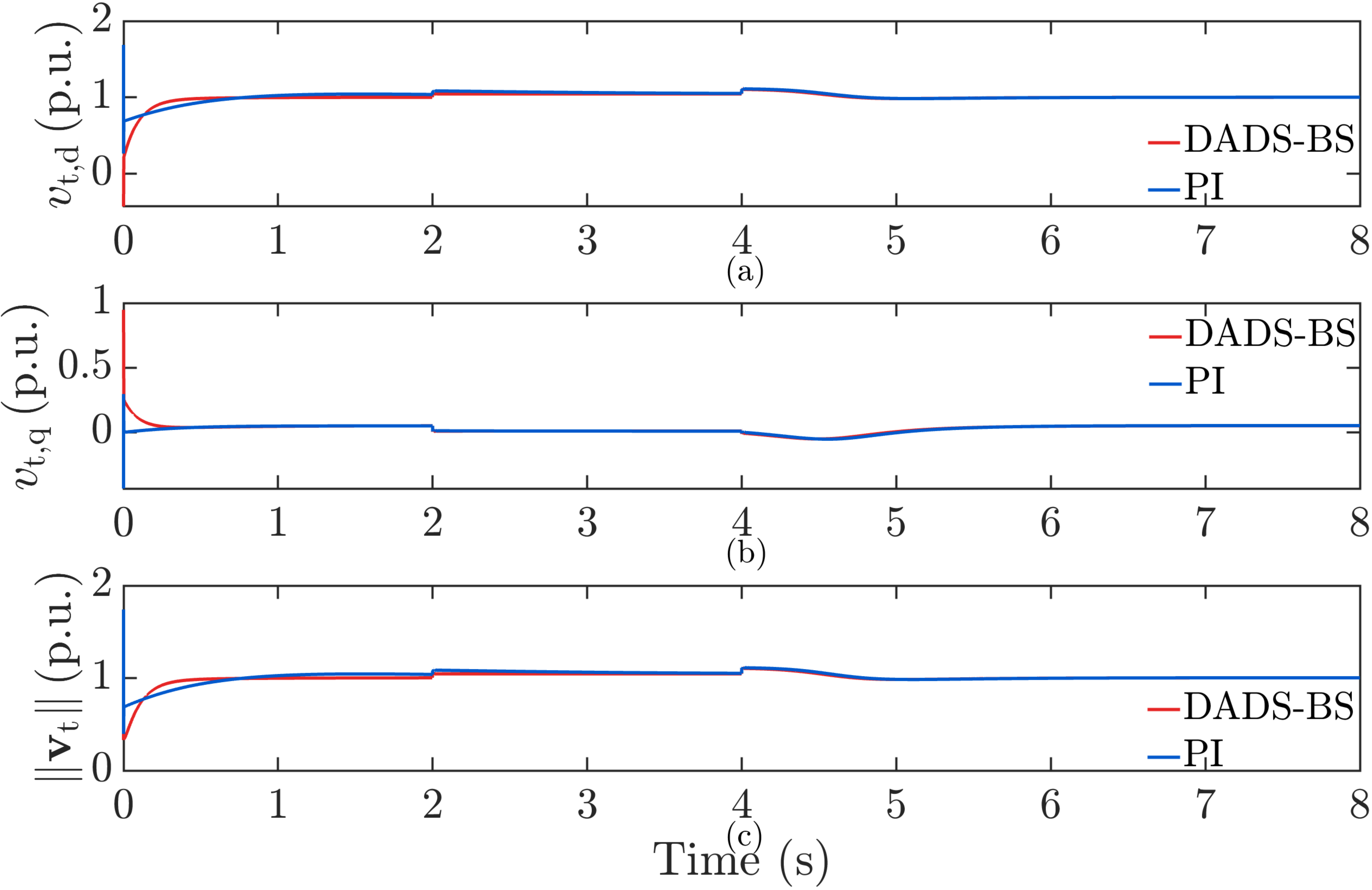}
\caption{Terminal voltage (control input) \textbf{without} current limiting. (a) $\rm d$ component. (b) $\rm q$ component. (c) $\boldsymbol{v}_{\rm t} = \sqrt{v_{\rm t,d}^2+v_{\rm t,q}^2}.$}
    \label{ip_nmnl}
\end{figure}

\begin{figure}
    \centering
    \includegraphics[width=1\linewidth]{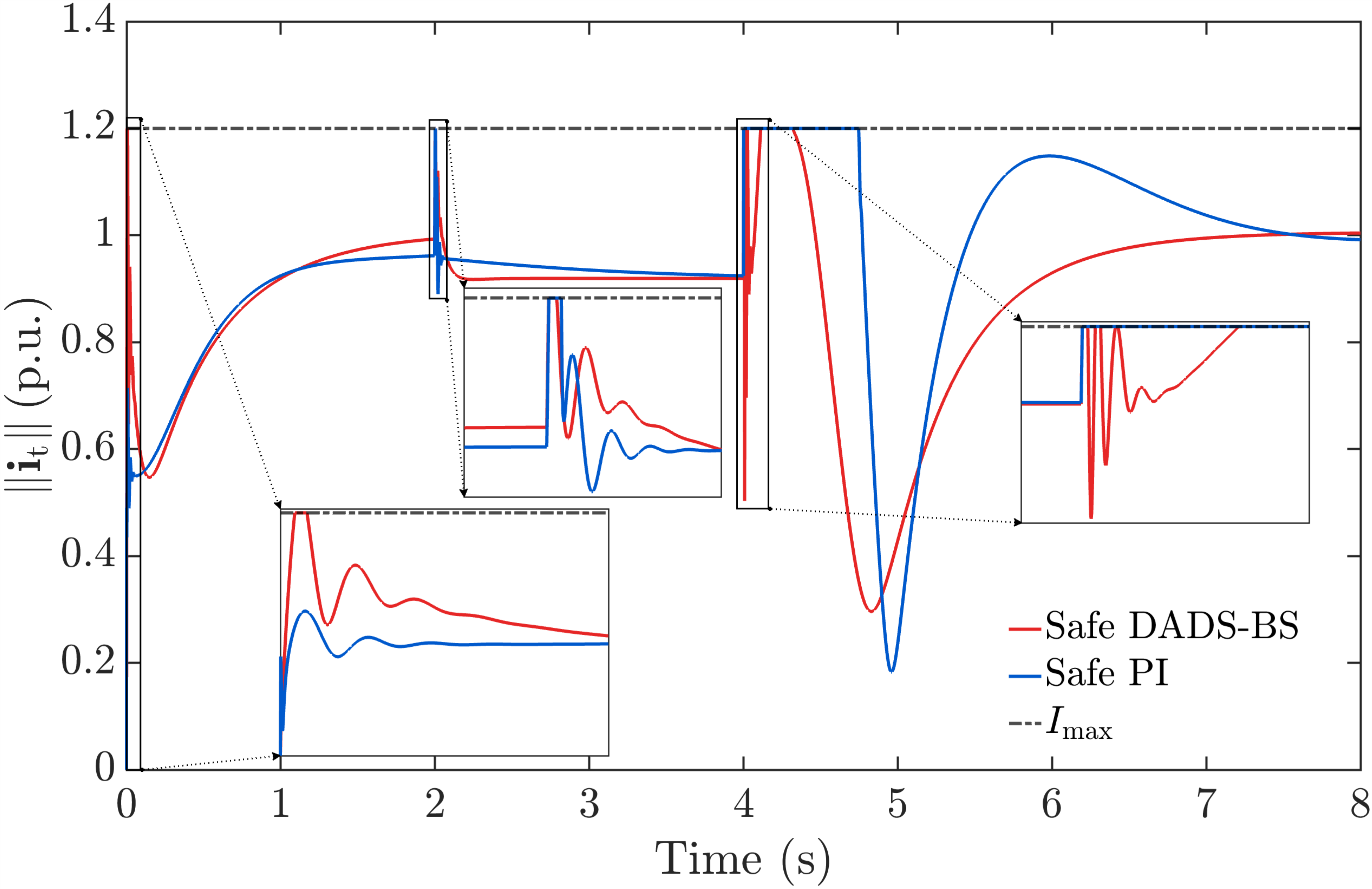}
\caption{Terminal current behavior \textbf{with} current limiting.}
    \label{termcur_sf}
\end{figure}

\begin{figure}
    \centering
    \includegraphics[width=1\linewidth]{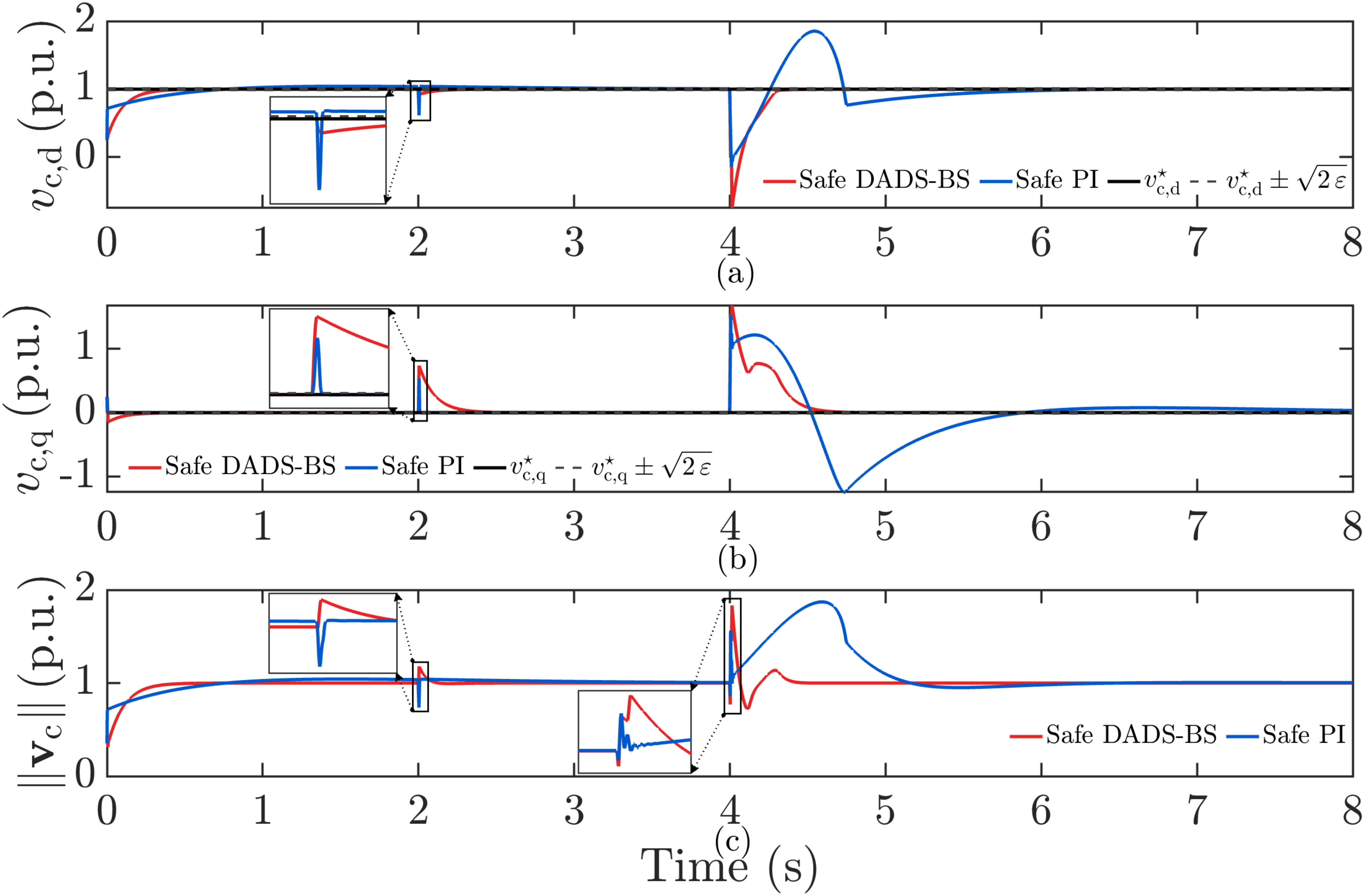}
\caption{PCC voltage behavior \textbf{with} current limiting. (a) $\rm d$ component. (b) $\rm q$ component. (c) $\boldsymbol{v}_{\rm c} = \sqrt{v_{\rm c,d}^2+v_{\rm c,q}^2}$.}
    \label{vcdq_sf}
\end{figure}

\begin{figure}
    \centering
    \includegraphics[width=1\linewidth]{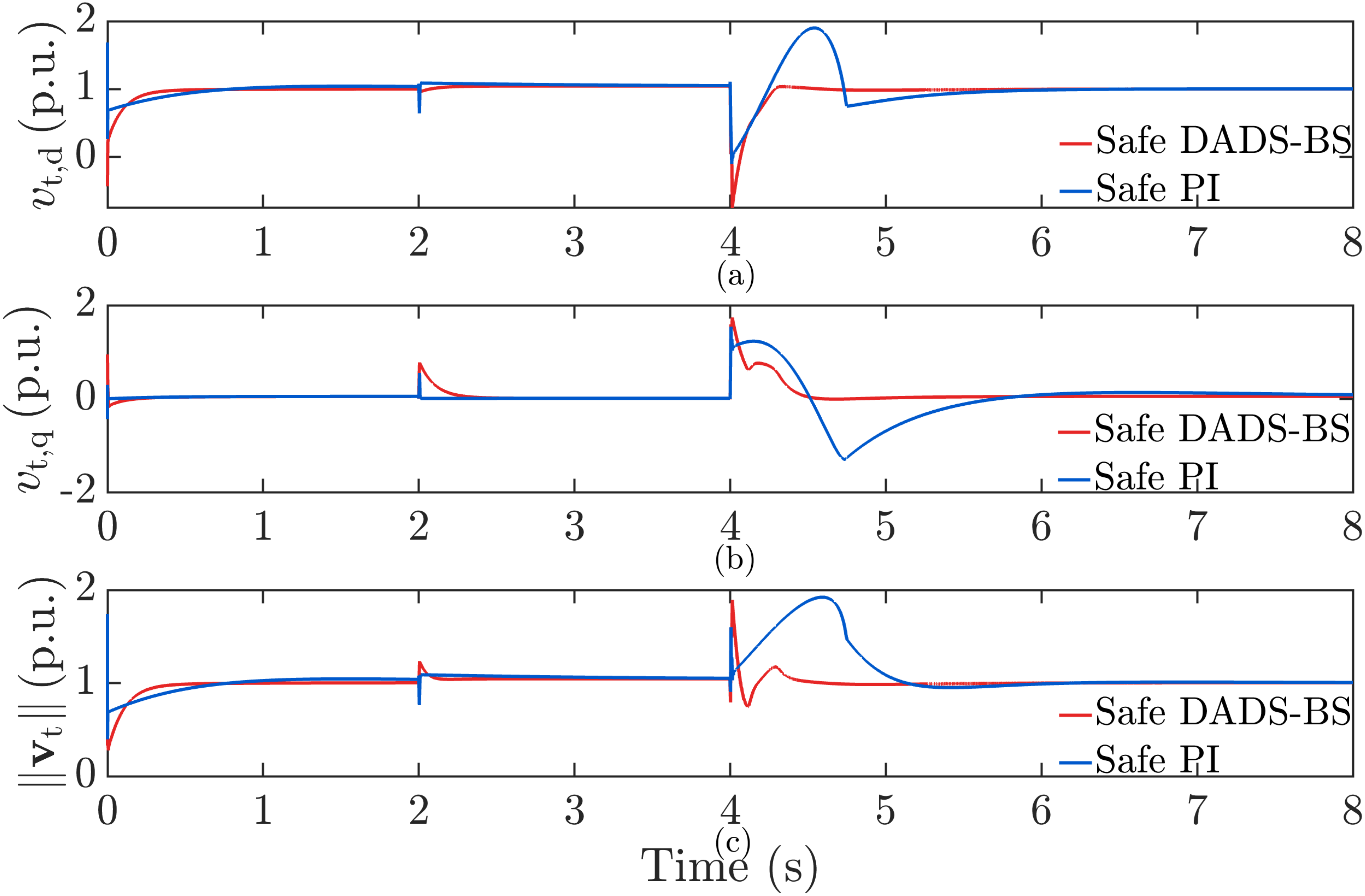}
\caption{Terminal voltage (control input) \textbf{with} current limiting. (a) $\rm d$ component. (b) $\rm q$ component. (c) $\boldsymbol{v}_{\rm t} = \sqrt{v_{\rm t,d}^2+v_{\rm t,q}^2}$.}
    \label{ip_sf}
\end{figure}

This section evaluates the proposed DADS-BS GFM controller against a conventional cascaded PI voltage--current control architecture \cite{chatterjee2025effects}, both \textbf{without} and \textbf{with} the proposed current-limiting mechanism. Specifically, we report two comparisons: (i) DADS-BS versus PI with the safety filter disabled (\textit{i.e.}, DADS-BS and PI are applied for all times without considering safety), and (ii) Safe DADS-BS versus Safe PI with the safety filter enabled, that is, using DADS-BS and PI as nominal controllers within the safety-filtered control law in \eqref{sfty_ip2},\eqref{sfty_ip1}.
The PI controller structure and tuned following \cite{chatterjee2025effects}.

The inverter and network parameters are selected in p.u. as $C_{\rm f} = 0.30,\,L_{\rm f} = 0.05,\,R_{\rm f} = 7.2\times 10^{-3},\, R = 0.2,\, L = 0.8$, and the base electrical frequency as $\omega_{\rm b} = 120\pi$ $\rm{rad}/s$. For the DADS-BS controller, the gains and adaptive-law parameters are $K_{\rm VC}=K_{\rm CC}=10 \, \rm{s^{-1}},\,\Gamma_{\rm d}=\Gamma_{\rm q}=10^6 \, \rm{s^{-1}},\,\mu_{\rm d}=\mu_{\rm q}=1\, \rm{s^{-1}},\,\varepsilon = 10^{-4}$. The droop parameters and setpoints are chosen in p.u. as $K_{\rm P} = 5\times 10^{-3},\,K_{\rm Q}=1\times 10^{-4},\,P_{0} = 1,\,Q_{0} = 0.5,\, \omega_{0} =1,\,V_{0} =1$, and the power filter parameters are chosen as $\omega_{\rm pc} = 332.8 \,\rm{rad}/s,\,\omega_{\rm qc} = 732.8 \,\rm{rad}/s,\, \xi_{\rm p} = 1.2,\,\xi_{\rm q} = 1.2.$ The simulations use the same initial conditions across all cases. For current limiting, the bound is set to $I_{\max}=1.2$ p.u., and the safety-filter gain is chosen as $c= 1\times 10^9 \, \rm{s^{-1}}$. All simulations are performed using a stiff solver (MATLAB \texttt{ode15s}) with tolerances $\texttt{RelTol}=10^{-7}$, $\texttt{AbsTol}=10^{-9}$, and $\texttt{MaxStep}=10^{-4}$.

Figure \ref{grdvol} shows the grid voltage, with Fig. \ref{grdvol}(a) depicting the three-phase $\rm abc$ voltages and Fig. \ref{grdvol}(b) showing the grid voltage in the global $\rm DQ$ frame. During $t\in[2,4]$s, a three-phase-to-ground fault is applied at the grid.

Figures \ref{vcdq_nmnl}–\ref{ip_nmnl} present the results \textbf{without} current limiting for both DADS-BS and PI control. As shown in Fig. \ref{vcdq_nmnl}, variations in the grid voltage do not affect DADS-BS due to its disturbance-suppression property, and convergence within the $\sqrt{2\varepsilon}$ bands is observed, as guaranteed by the theory. However, Fig. \ref{termcur_nmnl} shows that both DADS-BS and PI exceed the allowable terminal current limit $I_{\rm max}=1.2$, primarily triggered by the fault, highlighting the necessity of explicit current-limiting mechanisms. Figure \ref{ip_nmnl} shows the terminal voltages (control inputs) under both DADS-BS and PI.

Figures \ref{termcur_sf}–\ref{ip_sf} show the results \textbf{with} current limiting for both Safe DADS-BS and Safe PI. As seen in Fig. \ref{termcur_sf}, both controllers effectively enforce the current limit without violation. However, as a consequence, Fig. \ref{vcdq_sf} shows that both Safe DADS-BS and Safe PI temporarily deviate from their GFM objectives. This behavior reflects an inherent tradeoff between strict enforcement of hard current constraints and voltage regulation. When the safety constraint becomes active, the controller prioritizes current limitation, thereby temporarily relaxing voltage regulation. This deviation arises because the admissible control set is restricted by the CBF constraint during current-limiting events. Notably, Safe DADS-BS recovers rapidly during and after the fault and promptly resumes GFM behavior, demonstrating superior transient performance during and following current-limiting events. Figure \ref{ip_sf} shows the terminal voltages under both Safe DADS-BS and Safe PI.

\section{Proofs}
\label{secProof}
\subsection{Auxiliary Lemmas}
\begin{lem}\label{lemfilt}
Let $b>0$, $\xi>1$, and let $u\in L^\infty(\mathbb{R}_{+};\mathbb{R})$ with $\|u\|_\infty=M$. Consider the system
\begin{align}\label{eta_fltr1}
\dot\eta_{1}&=\eta_{2},\\ \label{eta_fltr2}
\dot\eta_{2}&=-2\xi b \eta_{2}(t)-b^2 \eta_{1}+b^2 u, 
\end{align}
for $t\ge 0$, with initial condition $(\eta_{1}(0),\eta_{2}(0))=(\eta_{10},\eta_{20})\in\mathbb{R}^2$.  
Then, there exists a unique  absolutely continuous solution $(\eta_{1},\eta_{2}):\mathbb{R}_{+} \to \mathbb{R}^2$ to \eqref{eta_fltr1},\eqref{eta_fltr2}, and it satisfies
\begin{align}\label{eta1_bounds_final}
|\eta_{1}(t)|
&\le \frac{\xi|\eta_{10}| + \frac{|\eta_{20}|}{b}}{\sqrt{\xi^2 - 1}} + M, \\
\label{eta2_bounds_final}
|\eta_{2}(t)|
&\le \frac{b|\eta_{10}| + \xi|\eta_{20}|}{\sqrt{\xi^2 - 1}} + \frac{b M}{\sqrt{\xi^2 - 1}},
\end{align}
for all $t \ge 0$.
\end{lem}

\noindent \textbf{Proof.} We rewrite the system in vector form as

\begin{align*}
    \dot{\boldsymbol{\eta}}  = \boldsymbol{H} \boldsymbol{\eta} + \boldsymbol{B}u,
\end{align*}
where
\begin{align*}
    \boldsymbol{\eta} = \begin{bmatrix}
        \eta_1\\\eta_2
    \end{bmatrix},\quad \boldsymbol{H} = \begin{bmatrix}
        0 && 1\\
        -b^2&& -2\xi b
    \end{bmatrix},\quad \boldsymbol{B} = \begin{bmatrix}
        0\\b^2
    \end{bmatrix}.
\end{align*}
Since $b>0$ and $\xi>1$, matrix $\boldsymbol{H}$ has two distinct real eigenvalues $(-\gamma_1,-\gamma_2)$ with $\gamma_1>\gamma_2>0$ given by 
\begin{align*}
\gamma_1=b\big(\xi+\sqrt{\xi^2-1}\big),\qquad \gamma_2=b\big(\xi-\sqrt{\xi^2-1}\big).
\end{align*}
Standard eigen-decomposition yields
\begin{align*}
    \boldsymbol{H} = \frac{1}{\gamma_1-\gamma_2}\begin{bmatrix}1 & 1\\ -\gamma_2 & -\gamma_1\end{bmatrix} \begin{bmatrix}
        -\gamma_2 & 0\\
        0 & -\gamma_1
    \end{bmatrix} \begin{bmatrix}
\gamma_1 & 1\\ -\gamma_2 & -1
\end{bmatrix},
\end{align*}
and therefore
\begin{align*}
    e^{\boldsymbol{H}t} = \frac{1}{\gamma_1-\gamma_2}\begin{bmatrix}1 & 1\\ -\gamma_2 & -\gamma_1\end{bmatrix} \begin{bmatrix}
        e^{-\gamma_2 t} & 0\\
        0 & e^{-\gamma_1 t} 
    \end{bmatrix} \begin{bmatrix}
\gamma_1 & 1\\ -\gamma_2 & -1
\end{bmatrix}.
\end{align*}
Using variation of constants, the full solution is
\begin{align*}
     \boldsymbol{\eta}(t) = e^{ \boldsymbol{H}t}  \boldsymbol{\eta}(0)+\int_{0}^{t} e^{\boldsymbol{H}(t-\tau)}\boldsymbol{B} u(\tau)d\tau.
\end{align*}
Evaluating the first component of $\eta$ gives
\begin{align}\label{int_eta11}
    \eta_1(t) \!= \!\eta_{1,\rm{h}}(t) \!+\! \frac{b^2}{\gamma_1-\gamma_2}\int_{0}^{t} \big(e^{-\gamma_2 (t-\tau)}\!-\!e^{-\gamma_1 (t-\tau)}\big)u(\tau)d\tau,
\end{align}
where $\eta_{1,\rm{h}}(t)$ is 
\begin{align}\label{eta12h}
\eta_{\rm 1,h}(t)=c_{1}e^{-\gamma_1 t}+c_{2}e^{-\gamma_2 t},
\end{align}
with 
\begin{align}\label{ceff}
c_{1}=\frac{-\gamma_2\eta_{10}-\eta_{20}}{\gamma_1-\gamma_2}, \qquad c_{2}=\frac{\gamma_1\eta_{10}+\eta_{20}}{\gamma_1-\gamma_2}.
\end{align}
Since $u$ is bounded with $\|u\|_\infty=M$, we estimate 
\begin{align*}
\begin{split}
    &\vert \eta_1(t)\vert  \leq \vert \eta_{1,\rm h}(t)\vert + \frac{b^2 M}{\gamma_1 -\gamma_2}\int_{0}^{t}\big\vert e^{-\gamma_2 (t-\tau)}-e^{-\gamma_1 (t-\tau)}\big\vert d\tau\\& = \vert \eta_{1,\rm h}(t)\vert + \frac{b^2 M}{\gamma_1 -\gamma_2}\int_{0}^{t}\big( e^{-\gamma_2 (t-\tau)}-e^{-\gamma_1 (t-\tau)}\big)d\tau\\&=\vert \eta_{1,\rm h}(t)\vert + \frac{b^2 M}{\gamma_1 -\gamma_2}\Big(\frac{1}{\gamma_2}-\frac{1}{\gamma_1}-\frac{1}{\gamma_2}e^{-\gamma_2 t}+\frac{1}{\gamma_1}e^{-\gamma_1 t}\Big).
\end{split}
\end{align*}
Noting that $e^{-\gamma_2t}/\gamma_2 > e^{-\gamma_1t}/\gamma_1$ for all $t\geq 0$ due to $\gamma_1>\gamma_2 >0$ gives
\begin{align}\label{eta1_fn1}
\begin{split}
|\eta_1(t)|
&\le |\eta_{\rm 1,h}(t)|
+\frac{b^2M}{\gamma_1-\gamma_2}\Big(\frac{1}{\gamma_2}-\frac{1}{\gamma_1}\Big)\\&
=|\eta_{\rm 1,h}(t)|+\frac{b^2M}{\gamma_2\gamma_1}=|\eta_{\rm 1,h}(t)|+M,
\end{split}
\end{align}
since $\gamma_2\gamma_1=b^2$. Using \eqref{eta12h} and \eqref{ceff},
\begin{align}\label{eta1_fn2}
|\eta_{1,h}(t)|\!\le \!|c_1|e^{-\gamma_1 t} \!+\! |c_2|e^{-\gamma_2 t}
\!\le\! |c_1|\!+\!|c_2|
\!\leq \!\frac{\xi|\eta_{10}|\!+\!|\eta_{20}|/b}{\sqrt{\xi^2\!-\!1}}.
\end{align}
Combining \eqref{eta1_fn1} and \eqref{eta1_fn2} gives \eqref{eta1_bounds_final}. 

Next, we compute $\eta_2(t)$.  From \eqref{eta_fltr1},\eqref{int_eta11},\eqref{eta12h}, 
\begin{align*}
     \eta_2(t) \!=\! \eta_{2,\rm{h}}(t) \!+\! \frac{b^2}{\gamma_1\!-\!\gamma_2}\!\int_{0}^{t} \!\big(\gamma_1 e^{-\gamma_1 (t-\tau)}\!-\!\gamma_2 e^{-\gamma_2 (t-\tau)}\big)u(\tau)d\tau,
\end{align*}
where  $\eta_{2,\rm{h}}(t)$ is 
\begin{align}\label{eta2nd}
\eta_{2,h}(t)=-\gamma_1 c_{1}e^{-\gamma_1 t}-\gamma_2 c_{2}e^{-\gamma_2 t}. 
\end{align}
Using $\|u\|_\infty=M$ gives 
\begin{align}\label{eta2_fnl1}
\begin{split}
    \vert \eta_2(t)\vert  &  \leq \vert \eta_{2,\rm{h}}(t)\vert \!+\! \frac{b^2 M}{\gamma_1-\gamma_2}\!\int_{0}^{t} \big(\gamma_2 e^{-\gamma_2 (t-\tau)}\!+\!\gamma_1 e^{-\gamma_1 (t-\tau)}\big)d\tau\\& = \vert \eta_{2,\rm{h}}(t)\vert+\frac{b^2 M}{\gamma_1-\gamma_2}\Big(2 - e^{-\gamma_2 t}-e^{-\gamma_1 t}\Big)\\&\leq \vert \eta_{2,\rm{h}}(t)\vert+\frac{2b^2 M}{\gamma_1-\gamma_2} = \vert \eta_{2,\rm{h}}(t)\vert+\frac{b M}{\sqrt{\xi^2-1}}.
\end{split}
\end{align}
Using \eqref{eta2nd} and \eqref{ceff},
\begin{align}\label{eta2_fnl2}
\begin{split}
 |\eta_{2,h}(t)|
\le&\gamma_1|c_1|e^{-\gamma_1 t}+ \gamma_2|c_2|e^{-\gamma_2 t}
\\\le& \gamma_1|c_1|+\gamma_2|c_2|
\leq \frac{b|\eta_{10}|+\xi|\eta_{20}|}{\sqrt{\xi^2-1}}.   
\end{split}
\end{align}
Combining \eqref{eta2_fnl1} and \eqref{eta2_fnl2} gives \eqref{eta2_bounds_final}.  This completes the proof.

\hfill $\qed$

\begin{lem}[\it Karafyllis Lemma on a Shifted Interval]\label{krflslm}
Let $T_1,T_2 \in \mathbb{R}$ with $0 \leq T_1 < T_2 \le +\infty$ be given. 
Consider two absolutely continuous functions 
$z,V : [T_1,T_2) \to \mathbb{R}$ for which the following differential 
inequalities hold for $t \in [T_1,T_2)$ a.e.:
\begin{align}
    0 \le \dot{z}(t) 
    &\le \Gamma e^{-z(t)} \max\{W(t)-\varepsilon, 0\}, \label{karaf_shift_z}\\
    \dot{W}(t) 
    &\le -c W(t) + \frac{\mu}{1+e^{z(t)}}, \label{karaf_shift_V}
\end{align}
for certain constants $c,\varepsilon,\Gamma >0$ and $\mu \ge 0$. 
Then the following estimate holds for all $t \in [T_1,T_2)$:
\begin{equation}\label{lemma41_bound_shifted}
\begin{split}
z(T_1)\le&\; z(t) \le 
\ln\Bigg(
  \max\Big(
    \frac{\mu}{c\varepsilon}-1, e^{z(T_1)}
  \Big) \\
  &\quad + \frac{\Gamma}{c}
  \max\Big\{
    W(T_1) + \frac{\mu}{c\big(1+e^{z(T_1)}\big)} - \varepsilon, 0
  \Big\}
\Bigg).
\end{split}
\end{equation}
\end{lem}

The proof follows similar arguments to those of Lemma 4.1 in \cite{karafyllis2025robust}.

\subsection{Proof of Theorem \ref{thm1}}\label{pfthm1} 
The right-hand side of the closed-loop dynamics \eqref{inv1}-\eqref{dads_qend} is locally Lipschitz in $(\boldsymbol{x},z_{\rm d},z_{\rm q})$. Further, from Assumption \ref{assdis}, the grid voltage components satisfy $v_{\rm g,d},v_{\rm g,q}\in L^\infty(\mathbb{R}_{+};\mathbb{R})$. Hence, for every initial condition 
$\big(\boldsymbol{x}(0),z_{\rm d}(0),z_{\rm q}(0)\big)\in\mathbb{R}^{10}\times\mathbb{R}\times\mathbb{R}$ there exists a unique absolutely continuous maximal solution on some interval $[0,t_{\max})$, where $t_{\rm \max}\in (0,+\infty]$. 

\noindent \textit{\underline{Boundedness of the filtered reactive powers:}}

We bound $q_1$ and $q_2$  using the second-order filter \eqref{Q_filter} with input $\operatorname{sat}_{\overline Q}(q)$ bounded by $\overline Q$. By Lemma \ref{lemfilt} with $b=\omega_{\rm qc}$ and $\xi=\xi_{\rm q}>1$,
\begin{align}\label{q1bar}
|q_1(t)|&\le \overline{q}_{\rm 1}(\boldsymbol{x}(0))
:=\frac{\xi_{\rm q}|q_1(0)|+\frac{|q_2(0)|}{\omega_{\rm qc}}}{\sqrt{\xi_{\rm q}^2-1}}+\overline Q,\\ \label{q2bar}
\vert q_2(t)\vert &\leq \overline{q}_{\rm 2}(\boldsymbol{x}(0))
:=  \frac{\omega_{\rm qc}|q_{1}(0)|+\xi_{\rm q}|q_{2}(0)|}{\sqrt{\xi^2_{\rm q}-1}}+\frac{\omega_{\rm qc} \overline{Q}}{\sqrt{\xi^2_{\rm q}-1}},
\end{align}
for all $t\in [0,t_{\rm max})$. 

\noindent \textit{\underline{$\rm d/q$-channel error dynamics and storage function bound:}} 

Define the $\rm d$-axis errors
\begin{align}\label{evdid}
e_{\rm v,d}:=v_{\rm c,d}-v^{\rm r}_{\rm c,d},\qquad
e_{\rm i,d}:=i_{\rm t,d}-i^{\rm r}_{\rm t,d}.
\end{align}

\noindent Using the plant equations, \eqref{Q_filter},\eqref{vcd_star}, the reference current $i_{\rm t,d}^{\rm r}$ \eqref{itd_r}, and \eqref{evdid}, we obtain
\begin{align}\label{evd}
\begin{split}
    &\dot{e}_{\rm v,d} 
    = \omega_{\rm b} \omega v_{\rm{c,q}} + \frac{\omega_{\rm b}}{C_{\rm f}}\big(i_{\rm{t,d}} - i_{\rm{g,d}}\big)+K_{\rm Q}q_2\\
    &=\omega_{\rm b} \omega v_{\rm{c,q}} + \frac{\omega_{\rm b}}{C_{\rm f}}\Big(e_{\rm i,d}+i_{\rm{g,d}}-C_{\rm f}\omega v_{\rm{c,q}}-\frac{C_{\rm f}K_{\rm Q}}{\omega_{\rm b}}q_2\\&\quad-\frac{C_{\rm f}K_{\rm VC}}{\omega_{\rm b}}\big(v_{\rm{c,d}}-v^{\rm r}_{\rm{c,d}}\big) - i_{\rm{g,d}}\Big)+K_{\rm Q}q_2\\
    &=\frac{\omega_{\rm b}}{C_{\rm f}}e_{\rm i,d}-K_{\rm VC}e_{\rm v,d}.
\end{split}
\end{align}

\noindent Differentiating $e_{\rm i,d}$ and using the plant equations, \eqref{Q_filter}-\eqref{omega_droop}, and \eqref{evd}, 
\begin{align*}
\begin{split}
    &\dot{e}_{\rm i,d} =2\omega_{\rm b} \omega \big(i_{\rm{t,q}} -i_{\rm{g,q}}\big)- \frac{\omega_{\rm b} R_{\rm f}}{L_{\rm f}} i_{\rm{t,d}} -\omega_{\rm b}\bigg(\frac{1}{L_{\rm f}}+\omega^2 C_{\rm f} \bigg) v_{\rm{c,d}} \\&\quad- C_{\rm f} K_{\rm P} p_2v_{\rm{c,q}}
      +K_{\rm VC}\big( i_{\rm{t,d}}-i^{\rm r}_{\rm{t,d}}\big)-\frac{C_{\rm f}K_{\rm VC}^2}{\omega_{\rm b}}\big( v_{\rm{c,d}}-v^{\rm r}_{\rm{c,d}}\big)\\&\quad-\frac{C_{\rm f}K_{\rm Q}}{\omega_{\rm b}}\Big(2\xi_{\rm q}\omega_{\rm {qc}}q_2 + \omega_{\rm {qc}}^2\big(q_1 - \operatorname{sat}_{\overline Q}(q)\big)\Big)
    + \frac{\omega_{\rm b}}{L_{\rm f}}v_{\rm{t,d}}\\&\quad+\frac{\omega_{\rm b} R}{L}i_{\rm{g,d}}-\frac{\omega_{\rm b}}{L}v_{\rm{c,d}}+\frac{\omega_{\rm b}}{L}v_{\rm{g,d}}.
\end{split}
\end{align*}
Choose $v_{\rm t,d}$ as \eqref{dads_d1}, from which it follows
\begin{align}\label{eid}
    \dot{e}_{\rm i,d} = u_{\rm d}+\frac{\omega_{\rm b} R}{L}i_{\rm{g,d}}-\frac{\omega_{\rm b}}{L}v_{\rm{c,d}}+\frac{\omega_{\rm b}}{L}v_{\rm{g,d}}.
\end{align}
Consider the storage function
\begin{align}\label{lyuap_d}
    W_{\rm d} = \frac{1}{2}e^2_{\rm v,d}+\frac{1}{2}e^2_{\rm i,d}.
\end{align}
Using \eqref{evd} and \eqref{eid},
\begin{align}\label{Wddot}
\begin{split}
       &\dot{W}_{\rm d}  = e_{\rm v,d}\dot{e}_{\rm v,d}+  e_{\rm i,d}\dot{e}_{\rm i,d}\\
       &=-K_{\rm VC}e^2_{\rm v,d}+ \frac{\omega_{\rm b}}{C_{\rm f}}e_{\rm v,d}e_{\rm i,d}+ e_{\rm i,d}u_{\rm d}+\frac{\omega_{\rm b} R}{L}e_{\rm i,d}i_{\rm{g,d}}\\&\quad-\frac{\omega_{\rm b}}{L}e_{\rm i,d}v_{\rm{c,d}} +\frac{\omega_{\rm b}}{L}e_{\rm i,d}v_{\rm{g,d}}.
\end{split}
\end{align}
Apply Young's inequality with any $\mu_{\rm d}>0$ and any $z_{\rm d}\in\mathbb{R}$:
\begin{align*}
    \frac{\omega_{\rm b} R}{L}e_{\rm i,d}i_{\rm{g,d}} &\leq \frac{(1+e^{z_{\rm d}})\omega_{\rm b}^2}{4\mu_{\rm d}} e^2_{\rm i,d} i^2_{\rm{g,d}}+\frac{\mu_{\rm d}R^2}{L^2(1+e^{z_{\rm d}})},\\
    -\frac{\omega_{\rm b}}{L}e_{\rm i,d}v_{\rm{c,d}}&\leq \frac{(1+e^{z_{\rm d}})\omega_{\rm b}^2}{4\mu_{\rm d}} e^2_{\rm i,d} v^2_{\rm{c,d}}+\frac{\mu_{\rm d}}{L^2(1+e^{z_{\rm d}})},\\
        \frac{\omega_{\rm b}}{L}e_{\rm i,d}v_{\rm{g,d}}&\leq \frac{(1+e^{z_{\rm d}})\omega_{\rm b}^2}{4\mu_{\rm d}}e^2_{\rm i,d}+\frac{\mu_{\rm d}v_{\rm{g,d}}^2}{L^2(1+e^{z_{\rm d}})}.
\end{align*}
Thus
\begin{align}\label{Wd_dot_pre_ud}
\begin{split}
&\dot W_{\rm d}\le
-K_{\rm VC}e_{\rm v,d}^2+\frac{\omega_{\rm b}}{C_{\rm f}}e_{\rm v,d}e_{\rm i,d}
+e_{\rm i,d}u_{\rm d}
\\&+\frac{(1+e^{z_{\rm d}})\omega_{\rm b}^2}{4\mu_{\rm d}}e_{\rm i,d}^2  \big(1+i_{\rm g,d}^2+v_{\rm c,d}^2\big)
+\frac{\mu_{\rm d}\big(1+R^2+v_{\rm g,d}^2\big)}{L^2(1+e^{z_{\rm d}})}.
\end{split}
\end{align}

\noindent Choose $u_{\rm d}$ as prescribed in \eqref{u_d}, which cancels the mixed term and dominates the gain-multiplied $e_{\rm i,d}^2$ part in \eqref{Wd_dot_pre_ud}. Therefore,

\begin{align}\label{Wd_dot_basic}
\begin{split}
\dot W_{\rm d}&\le -K_{\rm VC}e_{\rm v,d}^2-K_{\rm CC}e_{\rm i,d}^2
+\frac{\mu_{\rm d}\big(1+R^2+v_{\rm g,d}^2\big)}{L^2(1+e^{z_{\rm d}})}
\\&\le -2k W_{\rm d}
+\frac{\mu_{\rm d}\big(1+R^2+v_{\rm g,d}^2\big)}{L^2(1+e^{z_{\rm d}})},
\end{split}
\end{align}
where $ k = \min\{K_{\rm VC},K_{\rm CC}\}.$ The adaptation dynamics satisfy
\eqref{z_ddot}, hence $z_{\rm d}(t)$ is non-decreasing $\implies$ $z_{\rm d}(t)\ge z_{\rm d}(0)$ for all $t\in [0,t_{\rm max})$. Using \eqref{Wd_dot_basic} and $v_{\rm g,d}\in L^\infty$ from Assumption \ref{assdis},
\begin{align}
\begin{split}
W_{\rm d}(t)
&\le e^{-2kt}W_{\rm d}(0)
+\frac{\mu_{\rm d}}{L^2}\int_0^t e^{-2k(t-\tau)}
\frac{1+R^2+v_{\rm g,d}^2(\tau)}{1+e^{z_{\rm d}(\tau)}}d\tau \\
&\le e^{-2kt}W_{\rm d}(0)
+\frac{\mu_{\rm d}\big(1+R^2+\|v_{\rm g,d}\|_\infty^2\big)}{2kL^2\big(1+e^{z_{\rm d}(0)}\big)}\big(1-e^{-2kt}\big)\\
&\le e^{-2kt}W_{\rm d}(0)
+\frac{\mu_{\rm d}\big(1+R^2+\|v_{\rm g,d}\|_\infty^2\big)}{2kL^2\big(1+e^{z_{\rm d}(0)}\big)}.
\end{split}\label{Wd_bound_final}
\end{align}
Since $e_{\rm v,d}^2,e_{\rm i,d}^2\le 2 W_{\rm d}$, we get for all $t\in [0,t_{\rm max})$
\begin{align}\label{evdeid}
|e_{\rm v,d}(t)|,|e_{\rm i,d}(t)|
\le \sqrt{2W_{\rm d}(0)+\frac{\mu_{\rm d}\big(1+R^2+\|v_{\rm g,d}\|_\infty^2\big)}{kL^2\big(1+e^{z_{\rm d}(0)}\big)}}.
\end{align}
Similarly, for the $\rm{q}-$ channel, for error variables defined as 
\begin{align*}
e_{\rm v,q}:=v_{\rm c,q}-v^{\rm r}_{\rm c,q},\qquad
e_{\rm i,q}:=i_{\rm t,q}-i^{\rm r}_{\rm t,q},
\end{align*}
it holds that 
\begin{align*}
|e_{\rm v,q}(t)|,|e_{\rm i,q}(t)|
\le \sqrt{2W_{\rm q}(0)+\frac{\mu_{\rm q}\big(1+R^2+\|v_{\rm g,q}\|_\infty^2\big)}{kL^2\big(1+e^{z_{\rm q}(0)}\big)}},
\end{align*}
$t\in [0,t_{\rm max})$. 

\noindent \textit{\underline{Boundedness of the PCC voltage:}}

From the reactive power-voltage droop \eqref{vcd_star} and using \eqref{q1bar}, we have for $t\in [0,t_{\rm max})$,
\begin{align*}
|v^{\rm r}_{\rm c,d}(t)|
\le |V_0+K_{\rm Q}Q_{0}|+\vert K_{\rm Q}\vert \overline{q}_1.
\end{align*}
Further, noting that $v_{\rm c,d}=v^{\rm r}_{\rm c,d}+e_{\rm v,d}$ and  $|v_{\rm c,d}|\leq |v_{\rm c,d}^{\rm r}|+|e_{\rm v,d}|$, and combining with the error bound \eqref{evdeid}, we obtain 

\begin{align}\label{vcd_bnd}
\begin{split}
     \vert v_{\rm c,d}(t)\vert &\leq  \overline{v}_{\rm c,d}\big(\boldsymbol{x}(0),z_{\rm d}(0),\Vert v_{\rm g,d}\Vert_{\infty}\big)\\&\quad\quad :=\sqrt{2W_{\rm d}(0)+ \frac{\mu_{\rm d}\big(1+R^2+\Vert v_{\rm{g,d}}\Vert_{\infty}^2\big)}{kL^2\big(1+e^{z_{\rm d}(0)}\big)} }\\&\qquad\qquad+ \vert V_0+K_{\rm Q}Q_{0}\vert + \vert K_{\rm Q}\vert \overline{q}_1,
\end{split}
\end{align}
for all $t\in [0,t_{\rm max})$. Repeating similar steps, we obtain
\begin{align}
\begin{split}
\vert v_{\rm c,q}(t)\vert \leq& \overline{v}_{\rm c,q}\big(\boldsymbol{x}(0),z_{\rm q}(0),\Vert v_{\rm g,q}\Vert_{\infty}\big)\\&:= \sqrt{2W_{\rm q}(0)+ \frac{\mu_{\rm q}\big(1+R^2+\Vert v_{\rm{g,q}}\Vert_{\infty}^2\big)}{kL^2\big(1+e^{z_{\rm q}(0)}\big)} },
\end{split}\label{vcq_bnd}  
\end{align}
for all $t\in [0,t_{\rm max})$.

\noindent \textit{\underline{Boundedness of the grid current:}}

Let 
\begin{align*}
    W_{1}:=\frac{1}{2}(i_{\rm g,d}^2+i_{\rm g,q}^2).
\end{align*}
From the grid-inductor dynamics \eqref{invend1} and \eqref{invend},
\begin{align}\label{vgdot}
\begin{split}
\dot W_{1}
&= -\frac{\omega_{\rm b}R}{L}(i_{\rm g,d}^2+i_{\rm g,q}^2)
+\frac{\omega_{\rm b}}{L} i_{\rm g,d}(v_{\rm c,d}-v_{\rm g,d})
\\&\quad+\frac{\omega_{\rm b}}{L} i_{\rm g,q}(v_{\rm c,q}-v_{\rm g,q}). 
\end{split}
\end{align}
Applying Young's inequality repeatedly gives
\begin{align}
\begin{split}
    \label{yngineq1}
    \frac{\omega_{\rm b}}{L}i_{\rm g,d}\big(v_{\rm{c,d}}-v_{\rm{g,d}}\big)&\leq \frac{\omega_{\rm b}^2}{2L^2}\delta_1i^2_{\rm g,d}+\frac{1}{2\delta_1} \big(v_{\rm{c,d}}-v_{\rm{g,d}}\big)^2\\&\leq \frac{\omega_{\rm b}^2}{2L^2}\delta_1i^2_{\rm g,d}+\frac{1}{\delta_1} \big(v_{\rm{c,d}}^2+v_{\rm{g,d}}^2\big),
\end{split}
\end{align}
and
\begin{align}\label{yngineq2}
\begin{split}
        \frac{\omega_{\rm b}}{L}i_{\rm g,q}\big(v_{\rm{c,q}}-v_{\rm{g,q}}\big)&\leq \frac{\omega_{\rm b}^2}{2L^2}\delta_1 i^2_{\rm g,q}+\frac{1}{2\delta_1} \big(v_{\rm{c,q}}-v_{\rm{g,q}}\big)^2\\&\leq \frac{\omega_{\rm b}^2}{2L^2}\delta_1 i^2_{\rm g,q}+\frac{1}{\delta_1} \big(v_{\rm{c,q}}^2+v_{\rm{g,q}}^2\big),
\end{split}
\end{align}
for some $\delta_1>0$. Choosing $\delta_1 = \frac{LR}{\omega_{\rm b}}$, substituting \eqref{yngineq1} and \eqref{yngineq2} in \eqref{vgdot}, and since $i_{\rm g,d}^2+i_{\rm g,q}^2=2W_{1}$ and $v_{\rm g,d},v_{\rm g,q}\in L^{\infty}$,    
\begin{align*} 
\dot{W}_{1}&\le -\frac{\omega_{\rm b}R}{L}W_{1}
+\frac{\omega_{\rm b}}{LR}\big(v_{\rm c,d}^2+v_{\rm c,q}^2+v_{\rm g,d}^2+v_{\rm g,q}^2\big)\\&\le -\frac{\omega_{\rm b}R}{L}W_{1}
+\frac{\omega_{\rm b}}{LR}\big(\overline{v}_{\rm c,d}^2+\overline{v}_{\rm c,q}^2+\|v_{\rm g,d}\|_{\infty}^2+\|v_{\rm g,q}\|_{\infty}^{2}\big),
\end{align*}
where $\overline v_{\rm c,d},\overline v_{\rm c,q}>0$ are the bounds on $|v_{\rm c,d}|,|v_{\rm c,q}|$ given in \eqref{vcd_bnd} and \eqref{vcq_bnd}. From this, it follows 
\begin{align*}
W_{1}(t)
\le e^{-\frac{\omega_{\rm b}R}{L}t}W_{1}(0)
+\frac{\overline v_{\rm c,d}^2+\overline v_{\rm c,q}^2+\|v_{\rm g,d}\|_\infty^2+\|v_{\rm g,q}\|_\infty^2}{R^2},
\end{align*}
for all $t\in [0,t_{\rm max})$. Hence, noting that $i_{\rm g,d}^2,i_{\rm g,q}^2\leq 2W_{1}$, we obtain 
 \begin{align}
\begin{split}
    &\vert i_{\rm g,d}(t)\vert , \vert i_{\rm g,q}(t)\vert \leq \overline{i}_{\rm g}\big(\boldsymbol{x}(0),z_{\rm d}(0),z_{\rm q}(0),\Vert v_{\rm g,d}\Vert_{\infty},\Vert v_{\rm g,q}\Vert_{\infty}\big)\\&:=   \sqrt{i^2_{\rm g,d}(0)+i^2_{\rm g,q}(0) + \frac{2(\overline{v}_{\rm c,d}^2+\overline{v}_{\rm c,q}^2+\Vert v_{\rm g,d}\Vert^2_{\infty}+\Vert v_{\rm g,q}\Vert^2_{\infty})}{R^2}},\label{ig_bnd}
\end{split}
\end{align}
for all $t\in [0,t_{\rm max})$. 

\noindent \textit{\underline{Boundedness of the filtered active powers:}}

We bound $p_1$ and $p_2$  using the second-order filter \eqref{P_filter}. From \eqref{inst_act} and the bounds \eqref{vcd_bnd},\eqref{vcq_bnd},\eqref{ig_bnd}, it holds that $\vert p\vert\leq \overline{i}_{\rm g}(\overline{v}_{\rm c,d}+\overline{v}_{\rm c,q})$ for all $t\in[0,t_{\max})$. Since $|\operatorname{sat}_{\overline P}(y)|\le |y|$ for all $y\in\mathbb R$ and all $\overline P\in(0,+\infty]$, we get $
|\operatorname{sat}_{\overline P}(p(t))|\le \overline{i}_{\rm g}(\overline{v}_{\rm c,d}+\overline{v}_{\rm c,q}),\forall t\in[0,t_{\max})$. Then, by Lemma \ref{lemfilt} with $b=\omega_{\rm pc}$ and $\xi=\xi_{\rm p}>1$,
\begin{align}\label{p1bar}
\begin{split}
|p_1(t)|&\le \overline{p}_1\big(\boldsymbol{x}(0),z_{\rm d}(0),z_{\rm q}(0),\Vert v_{\rm g,d}\Vert_{\infty}, \Vert v_{\rm g,q}\Vert_{\infty}\big)
\\&\quad:=\frac{\xi_{\rm p}|p_1(0)|+\frac{|p_2(0)|}{\omega_{\rm pc}}}{\sqrt{\xi_{\rm p}^2-1}}+\overline{i}_{\rm g}(\overline{v}_{\rm c,d}+\overline{v}_{\rm c,q}),
\end{split}\\\label{p2bar}
\begin{split}
\vert p_2(t)\vert &\leq \overline{p}_2\big(\boldsymbol{x}(0),z_{\rm d}(0),z_{\rm q}(0),\Vert v_{\rm g,d}\Vert_{\infty}, \Vert v_{\rm g,q}\Vert_{\infty}\big)
\\&\quad:=  \frac{\omega_{\rm pc}|p_{1}(0)|+\xi_{\rm p}|p_{2}(0)|}{\sqrt{\xi^2_{\rm p}-1}}+\frac{\omega_{\rm pc} (\overline{i}_{\rm g}(\overline{v}_{\rm c,d}+\overline{v}_{\rm c,q}))}{\sqrt{\xi^2_{\rm p}-1}},
\end{split}
\end{align}
for all $t\in [0,t_{\rm max})$.

\noindent \textit{\underline{Boundedness of the terminal current:}}

Next, for $i_{\rm t,d}$, observe from \eqref{itd_r} that
\begin{align*}
|i^{\rm r}_{\rm t,d}|
\le |i_{\rm g,d}|+C_{\rm f}|\omega| |v_{\rm c,q}|
+\frac{C_{\rm f}\vert K_{\rm Q}\vert}{\omega_{\rm b}}|q_2|
+\frac{C_{\rm f}K_{\rm VC}}{\omega_{\rm b}}|e_{\rm v,d}|.
\end{align*}
The active power-frequency droop \eqref{omega_droop} and \eqref{p1bar} imply
\begin{align*}
|\omega|
\le \vert \omega_0 + K_{\rm P}P_{0}\vert +\vert K_{\rm P}\vert\overline{p}_1.
\end{align*}
Therefore, using \eqref{q2bar},\eqref{vcq_bnd},\eqref{ig_bnd},
\begin{align*}
\begin{split}
|i^{\rm r}_{\rm t,d}|
\le& \overline i_{\rm g}+C_{\rm f}\big(\vert \omega_0 + K_{\rm P}P_{0}\vert +\vert K_{\rm P}\vert\overline{p}_1\big) \overline v_{\rm c,q}
+\frac{C_{\rm f}\vert K_{\rm Q}\vert}{\omega_{\rm b}}\overline{q}_2
\\&+\frac{C_{\rm f}K_{\rm VC}}{\omega_{\rm b}}|e_{\rm v,d}|.
\end{split}
\end{align*}
Finally, $i_{\rm t,d}=i^{\rm r}_{\rm t,d}+e_{\rm i,d}$ together with the bound \eqref{evdeid} on $e_{\rm i,d}$ and $e_{\rm v,d}$  gives
\begin{align*}
\begin{split}
&|i_{\rm t,d}|
\le \overline i_{\rm g}+C_{\rm f}\big(\vert \omega_0 + K_{\rm P}P_{0}\vert +\vert K_{\rm P}\vert\overline{p}_1\big) \overline v_{\rm c,q}
+\frac{C_{\rm f}\vert K_{\rm Q}\vert}{\omega_{\rm b}}\overline{q}_2
\\&+\Big(\frac{C_{\rm f}K_{\rm VC}}{\omega_{\rm b}}+1\Big)\sqrt{2W_{\rm d}(0)+   \frac{\mu_{\rm d}\big(1+R^2+\Vert v_{\rm{g,d}}\Vert_{\infty}^2\big)}{kL^2\big(1+e^{z_{\rm d}(0)}\big)}},
\end{split}
\end{align*}
$t\in [0,t_{\rm max})$. Similar derivation on the $\rm q$-axis  yields 
    \begin{align*}
\begin{split}
    & \vert i_{\rm t,q}(t)\vert\leq   \overline i_{\rm g}+C_{\rm f}\big(\vert \omega_0 + K_{\rm P}P_{0}\vert +\vert K_{\rm P}\vert\overline{p}_1\big) \overline v_{\rm c,d}
\\&+\Big(\frac{C_{\rm f}K_{\rm VC}}{\omega_{\rm b}}+1\Big)\sqrt{2W_{\rm q}(0)+   \frac{\mu_{\rm q}\big(1+R^2+\Vert v_{\rm{g,q}}\Vert_{\infty}^2\big)}{kL^2\big(1+e^{z_{\rm q}(0)}\big)}},
\end{split}
\end{align*}
$t\in [0,t_{\rm max})$. 

\noindent \textit{\underline{Boundedness of the adaptive gains:}}

Considering \eqref{z_ddot},\eqref{dads_dend},\eqref{Wd_dot_basic}, recalling that $\vert v_{\rm g,d}(t)\vert\leq \Vert v_{\rm g,d}\Vert_{\infty}$ a.e, applying Lemma \ref{krflslm}, and noting that $z_{\rm d}$ is non-decreasing, we obtain 
\begin{align}\label{zdq_bnd}
\begin{split}
&z_{\rm d}(0)\\& \le z_{\rm d}(t) \le \ln\Bigg(
  \max\Big\{
    \frac{\mu_{\rm d}\big(1+R^2+\Vert v_{\rm g,d}\Vert^2_{\infty}\big)}{2kL^2\varepsilon} - 1,
    e^{z_{\rm d}(0)}
  \Big\} \\
 &+ \frac{\Gamma_{\rm d}}{2k}
  \max\Big\{
    W_{\rm d}(0) + \frac{\mu_{\rm d}\big(1+R^2+\Vert v_{\rm g,d}\Vert^2_{\infty}\big)}{2kL^2\big(1+e^{z_{\rm d}(0)}\big)} - \varepsilon,
    0
  \Big\}
\Bigg),
\end{split}
\end{align}
for all $t\in [0,t_{\rm max})$.  Similarly, we can obtain the upper bound for $z_{\rm q}$. 

Since all the states $(\boldsymbol{x}(t),z_{\rm d}(t),z_{\rm q}(t))$ are  bounded, $t_{\rm max} = +\infty$. Therefore, the unique absolutely continuous solution exists for all $t\geq 0$, and the bounds derived above hold for all $t\geq 0$.   This completes the proof. \hfill $\qed$

\subsection{Proof of Theorem \ref{thm2}}\label{pfthm2}

Consider the $\rm d$-channel storage function $W_{\rm d}$ in \eqref{lyuap_d}. Since $\big( v_{\rm c,d}-v_{\rm c,d}^{\rm r}\big)^2\leq 2 W_{\rm d}$ for all $t\geq 0$ and $1+e^{z_{\rm d}(0)}>1$, the estimate \eqref{re52} directly follows from \eqref{Wd_bound_final}. 

The  estimate \eqref{zdq_bnd} together with the fact that $z_{\rm d}$ is non-decreasing (note from \eqref{z_ddot} that $\dot z_{\rm d}(t)\ge 0$ a.e.,), implies that $z_{\rm d}$ has a finite limit as $t\to+\infty$. As a result, $e^{z_{\rm d}}$ also has a finite limit as $t \to +\infty$. Recall that $\dot W_{\rm d}$ is given by \eqref{Wddot}. There, $v_{\rm g,d}\in L^{\infty}(\mathbb{R}_{+};\mathbb{R})$; the terms $e_{\rm v,d},e_{\rm i,d}$ are bounded due to \eqref{evdeid}; the signals $i_{\rm g,d}$ and $v_{\rm c,d}$ are bounded by Theorem \ref{thm1}; and $u_{\rm d}$ in \eqref{u_d} is bounded because all states are bounded by Theorem \ref{thm1}. Therefore, $\dot W_{\rm d}$ is bounded. As a result,
\begin{align*}
\frac{d}{dt}e^{z_{\rm d}}=\Gamma_{\rm d}\max\{W_{\rm d}-\varepsilon,0\},
\end{align*}
is uniformly continuous. Furthermore, we have 
\begin{align*}
\int_0^{+\infty} \frac{d}{dt}e^{z_{\rm d}} dt = 
\lim_{t \to +\infty} \big(e^{z_{\rm d}} - e^{z_{\rm d}(0)}\big) < +\infty.
\end{align*}
Then, from Barbalat’s Lemma
\begin{align*}
\lim_{t\to+\infty}\left(\frac{d}{dt}e^{z_{\rm d}}\right)
=\Gamma_{\rm d}\lim_{t\to+\infty}\big(\max\{W_{\rm d}-\varepsilon,0\}\big)=0.
\end{align*}
Therefore,
\begin{align*}
\limsup_{t\to+\infty} W_{\rm d}\le \varepsilon,
\end{align*}
and hence $\limsup_{t\to+\infty}\lvert e_{\rm v,d}(t)\rvert\le \sqrt{2\varepsilon}$.

A verbatim repetition of the above steps for the $\rm q$–channel yields \eqref{re53} and \eqref{re55}. This completes the proof. \hfill $\qed$

\subsection{Proof of Theorem \ref{thm4}}\label{pfthm4}

The right-hand side of the closed-loop dynamics \eqref{inv1}–\eqref{xstte} with the control input
$\boldsymbol{v}_{\rm t} = [v_{\rm t,d},\, v_{\rm t,q}]^\top$ given by \eqref{sfty_ip1},\eqref{sfty_ip2},
where the nominal input $\boldsymbol{v}_{\rm t}^{\rm n} = [v_{\rm t,d}^{\rm n},\, v_{\rm t,q}^{\rm n}]^\top$
is any locally Lipschitz controller $\boldsymbol{k}(\boldsymbol{x})$ with $\boldsymbol{k}:\mathbb{R}^{10}\to\mathbb{R}^2$, is itself locally Lipschitz in the state. This is because, the safety filtered control \eqref{sfty_ip2},\eqref{sfty_ip1} is obtained by projecting
$\boldsymbol{v}_{\rm t}^{\rm n}$ onto a closed half-space defined by the affine CBF constraint,
which yields a locally Lipschitz mapping in $\boldsymbol{x}$ on the set
$\{\|\boldsymbol{i}_{\rm t}\|\neq 0\}$.
Moreover, when $\|\boldsymbol{i}_{\rm t}\| = 0$, all terms in \eqref{sfty_ip2} that depend on
$\boldsymbol{i}_{\rm t}$ vanish and $h(\boldsymbol x)=I_{\max}^2$, so that
$\eta(\boldsymbol x,\boldsymbol v_{\rm t}^{\rm n}) = c I_{\max}^2 > 0$.
By continuity of $\eta(\boldsymbol x,\boldsymbol v_{\rm t}^{\rm n})$, there exists $\epsilon>0$ such that
$\eta(\boldsymbol x,\boldsymbol v_{\rm t}^{\rm n}) \ge 0$ for all $\boldsymbol x$ satisfying
$\|\boldsymbol{i}_{\rm t}\| < \epsilon$.
Hence, in a neighborhood of the set $\{\|\boldsymbol{i}_{\rm t}\|=0\}$ the safety filter is inactive and
$\boldsymbol{v}_{\rm t} = \boldsymbol{v}_{\rm t}^{\rm n}$, which implies that
$\boldsymbol{v}_{\rm t}$ is locally Lipschitz in that neighborhood. Thus, the closed-loop vector field is locally Lipschitz in $\boldsymbol{x}$. Further, from Assumption~\ref{assdis}, the grid voltage components satisfy
$v_{\rm g,d},v_{\rm g,q}\in L^\infty(\mathbb{R}_{+};\mathbb{R})$.
Thus, for every initial condition $\boldsymbol{x}(0)\in\mathbb{R}^{10}$,
there exists a unique absolutely continuous maximal solution $\boldsymbol{x}(t)$
defined on some interval $[0,t_{\max})$, where $t_{\max}\in(0,+\infty]$.

\noindent \underline{\textit{Terminal current safety:}}

With the barrier function $h(\boldsymbol x)=I_{\max}^2-\|\boldsymbol i_{\rm t}\|^2$, the constraint in the QP \eqref{qpcnstnt} enforces $\dot h(\boldsymbol x(t))\ge -c h(\boldsymbol x(t))$ a.e for all $t\in [0,t_{\rm max})$. Thus, for all $t\in [0,t_{\rm max})$,
\begin{align*}
h(t)\ge e^{-ct}h(0)\;\Rightarrow\; h(t)\ge 0 \ \text{whenever}\ h(0)\ge 0.
\end{align*}
Equivalently, $\|\boldsymbol i_{\rm t}(t)\|\le I_{\max}$ for all $t\in [0,t_{\rm max})$.

\noindent \textit{\underline{Boundedness of the PCC voltage and grid current:}}

Let $\boldsymbol{v}_{\rm c} = [v_{\rm c,d},\, v_{\rm c,q}]^\top$,  $\boldsymbol{v}_{\rm g} = [v_{\rm g,d},\, v_{\rm g,q}]^{\top}$, and $\boldsymbol{i}_{\rm g} = [i_{\rm g,d},\, i_{\rm g,q}]^{\top}$. Concatenating the corresponding dynamics \eqref{inv1},\eqref{inv2} and \eqref{invend1},\eqref{invend}, we obtain  
\begin{align}
\dot{\boldsymbol{v}}_{\rm c} &= \omega_{\rm b}\omega\boldsymbol{J} \boldsymbol{v}_{\rm c} 
           + \frac{\omega_{\rm b}}{C_{\rm f}}\big(\boldsymbol{i}_{\rm t} - \boldsymbol{i}_{\rm g}\big), \label{vc_sys}\\
\dot{\boldsymbol{i}}_{\rm g} &= \omega_{\rm b}\omega\boldsymbol{J}\boldsymbol{i}_{\rm g}
           + \frac{\omega_{\rm b}}{L}\big(\boldsymbol{v}_{\rm c} - \boldsymbol{v}_{\rm g}\big)
           - \frac{\omega_{\rm b} R}{L}\boldsymbol{i}_{\rm g}\label{ig_sys},
\end{align}
where
\begin{align}\label{JJ}
    \boldsymbol{J} = \begin{bmatrix}
        0 && 1\\-1 && 0
    \end{bmatrix}. 
\end{align}
Consider the storage function 
\begin{align*}
    W_2 = \frac{1}{2}C_{\rm f}\Vert \boldsymbol{v}_{\rm c}\Vert^2+\frac{1}{2}L\Vert \boldsymbol{i}_{\rm g}\Vert^2 + \lambda \boldsymbol{v}_{\rm c}^{\top} \boldsymbol{i}_{\rm g}, \quad 0< \vert \lambda \vert < \sqrt{C_{\rm f} L}.
\end{align*}
with $\lambda$ to be chosen later. Completing the square gives
\begin{align*}
\begin{split}
    W_2 = &\frac{1}{2}C_{\rm f}\Big\Vert \boldsymbol{v}_{\rm c}+\frac{\lambda}{C_{\rm f}}\boldsymbol{i}_{\rm g}\Big\Vert^2+\frac{1}{2}\Big(L-\frac{\lambda^2}{C_{\rm f}}\Big)\Vert \boldsymbol{i}_{\rm g}\Vert^2,
\end{split}
\end{align*}
which is positive definite provided $    \vert \lambda \vert < \sqrt{C_{\rm f} L}.$ Using \eqref{vc_sys}-\eqref{JJ}, we obtain
\begin{align*}
\begin{split}
    \dot{W}_2 &= C_{\rm f}\boldsymbol{v}_{\rm c}^\top  \dot{\boldsymbol{v}}_{\rm c}+L\boldsymbol{i}_{\rm g}^\top  \dot{\boldsymbol{i}}_{\rm g} + \lambda \boldsymbol{v}_{\rm c}^{\top} \dot{\boldsymbol{i}}_{\rm g} + \lambda\boldsymbol{i}_{\rm g}^{\top} \dot{\boldsymbol{v}}_{\rm c}\\& = C_{\rm f}\boldsymbol{v}_{\rm c}^\top\Big(\omega_{\rm b}\omega\boldsymbol{J} \boldsymbol{v}_{\rm c} 
           + \frac{\omega_{\rm b}}{C_{\rm f}}\big(\boldsymbol{i}_{\rm t} - \boldsymbol{i}_{\rm g}\big)\Big)\\&\quad+L\boldsymbol{i}_{\rm g}^\top \Big(\omega_{\rm b}\omega\boldsymbol{J}\boldsymbol{i}_{\rm g}
           + \frac{\omega_{\rm b}}{L}\big(\boldsymbol{v}_{\rm c} - \boldsymbol{v}_{\rm g}\big)
           - \frac{\omega_{\rm b} R}{L}\boldsymbol{i}_{\rm g}\Big)\\&\quad+\lambda \boldsymbol{v}_{\rm c}^\top\Big(\omega_{\rm b}\omega\boldsymbol{J}\boldsymbol{i}_{\rm g}
           + \frac{\omega_{\rm b}}{L}\big(\boldsymbol{v}_{\rm c} - \boldsymbol{v}_{\rm g}\big)
           - \frac{\omega_{\rm b} R}{L}\boldsymbol{i}_{\rm g}\Big) \\&\quad+\lambda \boldsymbol{i}_{\rm g}^\top \Big(\omega_{\rm b}\omega\boldsymbol{J} \boldsymbol{v}_{\rm c} 
           + \frac{\omega_{\rm b}}{C_{\rm f}}\big(\boldsymbol{i}_{\rm t} - \boldsymbol{i}_{\rm g}\big)\Big),
\end{split}
\end{align*}
and, using the skew-symmetry of $\boldsymbol J$, we get
\begin{align*}
\begin{split}
    \dot{W}_2 =& \omega_{\rm b} \boldsymbol{v}_{\rm c}^\top \boldsymbol{i}_{\rm t}-\omega_{\rm b}\boldsymbol{i}_{\rm g}^\top \boldsymbol{v}_{\rm g}-\omega_{\rm b}R\Vert \boldsymbol{i}_{\rm g}\Vert^2\\&+\lambda \omega_{\rm b} \Big(\frac{1}{L}\Vert \boldsymbol{v}_{\rm c}\Vert^2-\frac{1}{L}\boldsymbol{v}_{\rm c}^\top \boldsymbol{v}_{\rm g}-\frac{R}{L}\boldsymbol{v}_{\rm c}^{\top}\boldsymbol{i}_{\rm g}+\frac{1}{C_{\rm f}}\boldsymbol{i}_{\rm g}^\top \boldsymbol{i}_{\rm t}\\&-\frac{1}{C_{\rm f}}\Vert \boldsymbol{i}_{\rm g}\Vert^2\Big).
\end{split}
\end{align*}
Applying Young's inequalities with constants  $\delta_2,\delta_3,\delta_4>0$,
\begin{align*}
\omega_{\rm b} \boldsymbol{v}_{\rm c}^\top \boldsymbol{i}_{\rm t} 
&\le \omega_{\rm b}\left(\frac{\delta_2}{2}\|\boldsymbol{v}_{\rm c}\|^2 + \frac{1}{2\delta_2}\|\boldsymbol{i}_{\rm t}\|^2\right),\\
-\omega_{\rm b} \boldsymbol{i}_{\rm g}^\top \boldsymbol{v}_{\rm g}
&\le \omega_{\rm b}\left(\frac{R}{2}\|\boldsymbol{i}_{\rm g}\|^2 + \frac{1}{2R}\|\boldsymbol{v}_{\rm g}\|^2\right),\\
-\lambda\omega_{\rm b} \frac{1}{L} \boldsymbol{v}_{\rm c}^\top \boldsymbol{v}_{\rm g}
&\le |\lambda|\omega_{\rm b}\left(\frac{1}{2L}\|\boldsymbol{v}_{\rm c}\|^2 + \frac{1}{2L}\|\boldsymbol{v}_{\rm g}\|^2\right),\\
-\lambda\omega_{\rm b} \frac{R}{L} \boldsymbol{v}_{\rm c}^\top \boldsymbol{i}_{\rm g}
&\le |\lambda|\omega_{\rm b}\frac{R}{L}\left(\frac{\delta_3}{2}\|\boldsymbol{v}_{\rm c}\|^2 + \frac{1}{2\delta_3}\|\boldsymbol{i}_{\rm g}\|^2\right),\\\lambda \omega_{\rm b} \frac{1}{C_f} \boldsymbol{i}_{\rm g}^\top \boldsymbol{i}_{\rm t}
&\le |\lambda|\omega_{\rm b}\left(\frac{\delta_4}{2C_{\rm f}}\|\boldsymbol{i}_{\rm g}\|^2 + \frac{1}{2\delta_4 C_{\rm f}}\|\boldsymbol{i}_{\rm t}\|^2\right).
\end{align*}
Hence,
\begin{align*}
\begin{split}
    \dot{W}_2 \leq& \omega_{\rm b}\Big(\frac{\delta_2}{2}+\frac{\lambda}{L}+\frac{\vert \lambda \vert}{2L}+\frac{\vert \lambda \vert R\delta_3}{2L}\Big)\Vert \boldsymbol{v}_{\rm c}\Vert^2\\&+\omega_{\rm b}\Big(-\frac{R}{2}+\frac{\vert \lambda\vert R}{2L\delta_3}+\frac{\vert \lambda\vert \delta_4}{2C_{\rm f}}+\frac{\vert \lambda\vert }{C_{\rm f}}\Big)\Vert \boldsymbol{i}_{\rm g}\Vert^2\\&+\omega_{\rm b}\Big(\frac{1}{2\delta_2}+\frac{\vert \lambda\vert}{2\delta_4 C_{\rm f}}\Big) \Vert \boldsymbol{i}_{\rm t}\Vert^2+\omega_{\rm b}\Big(\frac{1}{2R}+\frac{\vert \lambda\vert }{2L}\Big)\Vert \boldsymbol{v}_{\rm g}\Vert^2.
\end{split}
\end{align*}
Let
$\delta_2 = \frac{|\lambda|}{4L}, 
\delta_3 = \frac{1}{2R}, 
\delta_4 =1,$
which yields
\begin{align*}
\begin{split}
    \dot{W}_2 \leq& \omega_{\rm b} \bigg(\frac{8\lambda+7\vert \lambda\vert}{8L}\bigg)\Vert \boldsymbol{v}_{\rm c}\Vert^2\\&+\omega_{\rm b}\bigg(-\frac{R}{2}+\vert \lambda\vert\Big(\frac{R^2}{L}+\frac{3}{2C_{\rm f}}\Big) \bigg)\Vert \boldsymbol{i}_{\rm g}\Vert^2\\&+\omega_{\rm b}\bigg(\frac{2L}{\vert \lambda\vert}+\frac{\vert \lambda\vert}{2C_{\rm f}}\bigg)\Vert \boldsymbol{i}_{\rm t}\Vert^2 + \omega_{\rm b}\bigg(\frac{1}{2R}+\frac{\vert \lambda\vert}{2L}\bigg)\Vert \boldsymbol{v}_{\rm g}\Vert^2.
\end{split}
\end{align*}
Choose $\lambda<0$ such that
\begin{align*}
    -\min\Bigg\{\sqrt{C_{\rm f}L}, \frac{R}{4\Big(\frac{R^2}{L}+\frac{3}{2C_{\rm f}}\Big)} \Bigg\}<\lambda <0.
\end{align*}
With this, the $\Vert \boldsymbol v_{\rm c}\Vert^2$ and $\Vert \boldsymbol i_{\rm g}\Vert^2$ coefficients are strictly negative, and we obtain
\begin{align*}
    \dot W_2 \le - k_{1}\big(\|\boldsymbol v_{\rm c}\|^2+\|\boldsymbol i_{\rm g}\|^2\big)+k_{2}\|\boldsymbol i_{\rm t}\|^2+k_{3}\|\boldsymbol v_{\rm g}\|^2, 
\end{align*}
where $ k_{1} =\min\big\{\frac{\omega_{\rm b}\vert \lambda\vert}{8L},\frac{\omega_{\rm b}R}{4}\big\}$, $k_{2} =\omega_{\rm b}\Big(\frac{2L}{\vert \lambda\vert}+\frac{\vert \lambda\vert}{2C_{\rm f}}\Big)$, and $k_{3} = \omega_{\rm b}\Big(\frac{1}{2R}+\frac{\vert \lambda\vert}{2L}\Big)$. Then, recall that $\|\boldsymbol i_{\rm t}(t)\|\le I_{\rm \max}$ for all $t\in [0,t_{\rm max})$. By Assumption \ref{assdis}, $\|\boldsymbol v_{\rm g}\|$ is essentially bounded. Furthermore, since $W_2$ is positive definite in $(\Vert \boldsymbol{v}_{\rm c}\Vert,\Vert \boldsymbol{i}_{\rm g}\Vert)$, there exists $\overline{c}>\underline{c}>0$ such that \begin{align}\label{pstd2}
    \underline{c}\big(\|\boldsymbol v_{\rm c}\|^2+\|\boldsymbol i_{\rm g}\|^2\big) \leq W_2 \leq \overline{c}\big(\|\boldsymbol v_{\rm c}\|^2+\|\boldsymbol i_{\rm g}\|^2\big).
\end{align} Therefore, noting that $\Vert \boldsymbol{v}_{\rm g}\Vert_{\infty}^2\leq \Vert v_{\rm g,d}\Vert_{\infty}^2 + \Vert v_{\rm g,q}\Vert_{\infty}^2$, 
\begin{align*}
    \dot W_2(t) \le -k_{4} W_2(t) + k_{2} I_{\rm max}^2 + k_{3}(\Vert v_{\rm g,d}\Vert_{\infty}^2+\Vert v_{\rm g,q}\Vert_{\infty}^2),
\end{align*}
for $t\in [0,t_{\max})$ a.e., where $k_{4} = k_{1}/\overline{c}>0$. By comparison,
\begin{align}\label{W6est}
    W_2(t)\!\le\! e^{-k_{4} t}W_2(0)+\frac{k_{2}}{k_{4}}I_{\rm max}^2+\frac{k_{3}}{k_{4}}(\Vert v_{\rm g,d}\Vert_{\infty}^2\!+\!\Vert v_{\rm g,q}\Vert_{\infty}^2),
\end{align}
for all $t\in [0,t_{\rm max})$. Again considering \eqref{pstd2}, we deduce from \eqref{W6est} that there exists a function $\tilde B_{2,\boldsymbol \Omega}\in C^0(\mathbb R^2\times\mathbb R^2\times \mathbb R_+\times\mathbb R_+)$ such that
\begin{align} \label{vctrbnd}
\|[\boldsymbol v^\top_{\rm c}(t),\boldsymbol i^\top_{\rm g}(t)]^\top\|\le \tilde B_{2,\boldsymbol \Omega}\big(\boldsymbol v_{\rm c}(0),\boldsymbol i_{\rm g}(0),\|v_{\rm g,d}\|_\infty,\| v_{\rm g,q}\|_\infty\big),
\end{align}
for all $t\in [0,t_{\rm max})$.

\noindent \textit{\underline{Boundedness of the power filters:}}

The instantaneous powers satisfy $p=\boldsymbol v_{\rm c}^\top \boldsymbol i_{\rm g}$ and $q=\boldsymbol v_{\rm c}^\top \boldsymbol J^\top \boldsymbol i_{\rm g}$, so there exists $c_{\rm p},c_{\rm q}>0$ such that $|p(t)|\leq c_{\rm p}\|[\boldsymbol v^\top_{\rm c}(t),\boldsymbol i^\top_{\rm g}(t)]^\top\|^2$ and $|q(t)|\leq c_{\rm q}\|[\boldsymbol v^\top_{\rm c}(t),\boldsymbol i^\top_{\rm g}(t)]^\top\|^2$, hence it follows from \eqref{vctrbnd} that $p,q\in L^\infty$. Since $|\operatorname{sat}_{\overline P}(p)|\le |p|$ and $|\operatorname{sat}_{\overline Q}(q)|\le |q|$ for all $p,q\in\mathbb R$ and all $\overline P, \overline Q\in(0,+\infty]$, we get $
|\operatorname{sat}_{\overline P}(p(t))|\le \Vert p\Vert_{\infty}$ and $
|\operatorname{sat}_{\overline Q}(q(t))|\le \Vert q\Vert_{\infty}$ for all $t\in[0,t_{\max})$. Then, using Lemma \ref{lemfilt}, 
\begin{align*}
|p_1(t)|&\le \frac{\xi_{\rm p}|p_1(0)|+\frac{|p_2(0)|}{\omega_{\rm pc}}}{\sqrt{\xi_{\rm p}^2-1}}+\Vert p\Vert_{\infty},\\ 
|q_1(t)|&\le \frac{\xi_{\rm q}|q_1(0)|+\frac{|q_2(0)|}{\omega_{\rm qc}}}{\sqrt{\xi_{\rm q}^2-1}}+\Vert q\Vert_{\infty},\\\vert p_2(t)\vert &\leq   \frac{\omega_{\rm pc}|p_{1}(0)|+\xi_{\rm p}|p_{2}(0)|}{\sqrt{\xi^2_{\rm p}-1}}+\frac{\omega_{\rm pc} \Vert p\Vert_{\infty}}{\sqrt{\xi^2_{\rm p}-1}}, \\
\vert q_2(t)\vert &\leq   \frac{\omega_{\rm qc}|q_{1}(0)|+\xi_{\rm q}|q_{2}(0)|}{\sqrt{\xi^2_{\rm q}-1}}+\frac{\omega_{\rm qc} \Vert q\Vert_{\infty}}{\sqrt{\xi^2_{\rm q}-1}},
\end{align*}
for all $t\in [0,t_{\rm max})$. 

Since all the components of $\boldsymbol{x}(t)$ are bounded, we have that $t_{\rm max} = +\infty$. Therefore, the unique absolutely continuous solution exists for all $t\geq 0$, and the bounds derived above hold for all $t\geq 0$.   This completes the proof.  \hfill $\qed$

\subsection{Proof of Proposition \ref{propDADS}} \label{pfprp1}

 Following similar arguments to those in the proof of Theorem \ref{thm4}, it can be shown that for every initial condition $\big(\boldsymbol{x}(0),z_{\rm d}(0),z_{\rm q}(0)\big)\in\mathbb{R}^{10}\times \mathbb{R}\times \mathbb{R}$, there exists a unique maximal solution $\big(\boldsymbol{x}(t),z_{\rm d},z_{\rm q}\big)$
defined on some interval $[0,t_{\max})$, where $t_{\max}\in(0,+\infty]$. Furthermore, similar to as shown in the proof of Theorem \ref{thm4}, it can be shown that $i_{\rm t,d},i_{\rm t,q}$ and $\boldsymbol{x}'$ satisfy \eqref{crrntsft},\eqref{xdshbnd} on $[0,t_{\max})$.  Therefore, $\boldsymbol{x}$ is bounded on $[0,t_{\max})$. This implies the boundedness of $W_{\rm d}$ and $W_{\rm q}$ on $[0,t_{\max})$ that enter the
adaptation laws for $z_{\rm d}$ and $z_{\rm q}$ (see \eqref{vcd_star},\eqref{omega_droop},\eqref{itd_r},\eqref{z_ddot},\eqref{dads_dend},\eqref{itq_r},\eqref{z_qdot},\eqref{dads_qend}). Considering the $\rm d$ channel, since $\max\{W_{\rm d}-\varepsilon,0\}\ge 0$, we have $\dot z_{\rm d}(t)\ge 0$ a.e. on $[0,t_{\max})$, so $z_{\rm d}$ is
non-decreasing and $z_{\rm d}(t)\ge z_{\rm d}(0)$. Moreover, $
\frac{d}{dt}e^{z_{\rm d}}=e^{z_{\rm d}}\dot z_{\rm d}(t)=\Gamma_{\rm d}\max\{W_{\rm d}-\varepsilon,0\}.
$ Since $W_{\rm d}\in L^\infty([0,t_{\max});\mathbb{R})$, then $\frac{d}{dt}e^{z_{\rm d}}$ is essentially bounded on
$[0,t_{\max})$, and therefore $e^{z_{\rm d}(t)}$ (hence $z_{\rm d}(t)$) cannot blow up in finite time. An analogous argument holds for $z_{\rm q}$. Therefore, $t_{\rm max} = +\infty$ and the closed-loop system has a unique absolutely continuous solution for all $t\geq 0$.

\subsection{Proof of Theorem \ref{thm6}}\label{pfthm5}
In the proof of Proposition \ref{propDADS}, it is established that the bounds \eqref{crrntsft},\eqref{xdshbnd} hold on some interval $[0,t_{\max})$, where $t_{\max}\in(0,+\infty]$. Since the unique absolutely continuous solution of the closed-loop system exists for all $t\geq 0$ (see Proposition \ref{propDADS}), we conclude that the bounds \eqref{crrntsft},\eqref{xdshbnd} hold for all $t\geq 0$. 

Let us now proceed to establish \eqref{zqdbndsf}. For brevity, we consider the $\rm d-$ component only. Since $\boldsymbol{x}$ is bounded as in \eqref{crrntsft},\eqref{xdshbnd}, there exists a continuous function
\begin{align}\label{W_dbnd}
  \overline{W}_{\rm d} = \overline{W}_{\rm d}\big(\boldsymbol x(0),\| v_{\rm g,d}\|_\infty,\| v_{\rm g,q}\|_\infty\big),
\end{align}
such that
\begin{align}\label{Wd_global_bound}
  0\le W_{\rm d}(t)\le \overline{W}_{\rm d},
\end{align}
for all $t\geq 0$, where $W_{\rm d}$ is given by \eqref{dads_dend}. 

We first estimate the growth of $z_{\rm d}$ that can be attributed to the ON episodes $\mathcal I_{\rm on}$ \eqref{act_intvl}. From \eqref{z_ddot}, $\dot z_{\rm d}(t)\geq 0$ a.e., and hence $z_{\rm d}$ is non-decreasing. Therefore, using \eqref{W_dbnd} and the fact that $z_{\rm d}(t)\ge z_{\rm d}(0)$ for all $t\ge 0$, we obtain from \eqref{z_ddot}, for all $t\in\mathcal I_{\rm on}$ a.e.,
\begin{align*}
  \dot z_{\rm d}(t)\le \Gamma_{\rm d}e^{-z_{\rm d}(0)}\max\{\overline W_{\rm d}-\varepsilon,0\}.
\end{align*}
Integrating over the ON times gives the bound
\begin{align}\label{Delta_z_active}
\begin{split}
  \int_{\mathcal I_{\rm on}} \dot z_{\rm d}(s)ds
  &\le \Gamma_{\rm d}e^{-z_{\rm d}(0)}\max\{\overline W_{\rm d}-\varepsilon,0\}\int_{\mathcal{I}_{\rm on}}1ds
  \\&\le \Gamma_{\rm d}e^{-z_{\rm d}(0)}\max\{\overline W_{\rm d}-\varepsilon,0\}T_{\boldsymbol \Omega}
  =: \Delta z_{\rm d}^{\rm on}.
\end{split}
\end{align}
Above, we have used from Assumption \ref{asm:eta_events} that
$\int_{\mathcal I_{\rm on}}1\,ds = T_{\eta}\leq T_{\boldsymbol \Omega}$. Note that $\Delta z_{\rm d}^{\rm on}$ depends continuously on $\boldsymbol x(0),z_{\rm d}(0),\| v_{\rm g,d}\|_\infty,\| v_{\rm g,q}\|_\infty$, and linearly on $T_{\boldsymbol \Omega}$.

Next we focus on the behavior of $z_{\rm d}$ in OFF episodes. Define the OFF set
\begin{align*}
\mathcal I_{\rm off} := [0,+\infty)\setminus \mathcal I_{\rm on}= \{t\ge 0: \eta(t)\ge 0\}.
\end{align*}
Since $\mathcal I_{\rm on}$ is a finite union (at most $N_{\boldsymbol{\Omega}}$) of disjoint intervals of the form $[0,\sigma_0)$ (if $\eta(0)<0$) and $(\tau_\kappa,\sigma_\kappa)$, $\kappa\in\mathcal J$,  $\mathcal I_{\rm off}$ is a finite union of disjoint closed intervals, with the last one unbounded on the
right (see Remark \ref{rmksfz}). That is, there exist $N_{\rm off}\in\mathbb N\cup\{0\}$ and points
\begin{align*}
0 \!\leq \!s_0 \!\leq \!t_0 \!\leq\! s_1 \!\leq \!t_1\! \le \!\cdots \!\leq\! s_{N_{\rm off}-1}\! \leq\! t_{N_{\rm off}-1} \!\leq s_{N_{\rm off}} \!<\! +\infty,
\end{align*}
such that
\begin{align*}
  \mathcal I_{\rm off} = \bigcup_{\kappa=0}^{N_{\rm off}-1} [s_\kappa,t_\kappa] \cup [s_{N_{\rm off}},+\infty),
\end{align*}
where each $[s_\kappa,t_\kappa]$ is a closed interval
($s_\kappa=t_\kappa$ is allowed). Moreover, since $\mathcal I_{\rm on}$ has $N_\eta$ intervals, $\mathcal I_{\rm off}$ has at most $N_\eta+1$ intervals, and thus
\begin{align}\label{M_bound}
  N_{\rm off} \le N_\eta \le N_{\boldsymbol \Omega}.
\end{align}
On each OFF interval $J_\kappa=[s_\kappa,t_\kappa]\subseteq\mathcal I_{\rm off}$, where $\kappa=0,\dots,N_{\rm off}-1$, the controller coincides with the DADS-BS nominal controller. Therefore, similar to that it the proof of Theorem \ref{thm1}, we can show
\begin{align*}
  \dot W_{\rm d}(t) 
  \le -2k W_{\rm d}(t) 
  + \frac{\mu_{\rm d}^{\rm eff}}{1+e^{z_{\rm d}(t)}},
\end{align*}
 on $J_\kappa$ a.e., where 
\begin{align*}
\mu_{\rm d}^{\rm eff} := \frac{\mu_{\rm d}(1+R^2+\|v_{\rm g,d}\|_\infty^2)}{L^2},
\end{align*}
and $k$ is given by \eqref{decay_k}. Together with \eqref{z_ddot}, this matches the hypotheses of Lemma \ref{krflslm} on $[T_1,T_2)=[s_\kappa,t_\kappa)$. If $t_\kappa>s_\kappa$, Lemma \ref{krflslm} yields, for
all $t\in J_\kappa$,
\begin{align}\label{z_bound_interval}
\begin{split}
z_{\rm d}(s_\kappa)\le&\; z_{\rm d}(t) \\
\le&\; \ln\Bigg(
  \max\Big(
    \frac{\mu_{\rm d}^{\rm eff}}{2k\varepsilon}-1, e^{z_{\rm d}(s_\kappa)}
  \Big) \\
  &+ \frac{\Gamma_{\rm d}}{2k}
  \max\Big\{
    W_{\rm d}(s_\kappa) + \frac{\mu_{\rm d}^{\rm eff}}{2k(1+e^{z_{\rm d}(s_\kappa)})} - \varepsilon, 0
  \Big\}
\Bigg).
\end{split}
\end{align}
If $t_\kappa=s_\kappa$, then $J_\kappa$ reduces to a single point and
$z_{\rm d}(t_\kappa)=z_{\rm d}(s_\kappa)$, so \eqref{z_bound_interval} holds trivially.

With \eqref{Delta_z_active} and \eqref{z_bound_interval} at hand, below we work on obtaining an ultimate bound on $z_{\rm d}(t)$. Define the continuous function $\Phi_{\rm d}:\mathbb R\to\mathbb R$ by
\begin{equation}\label{Phi_d_def}
\begin{split}
\Phi_{\rm d}(\zeta)
:=& \ln\Bigg(
  \max\Big(
    \frac{\mu_{\rm d}^{\rm eff}}{2k\varepsilon}-1, e^{\zeta}
  \Big) 
  \\&\quad\quad+ \frac{\Gamma_{\rm d}}{2k}
  \max\Big\{
    \overline W_{\rm d} + \frac{\mu_{\rm d}^{\rm eff}}{2k} - \varepsilon, 0
  \Big\}
\Bigg).
\end{split}
\end{equation}
Note that, by construction of $\Phi_{\rm d}$,
\begin{align*}
  \Phi_{\rm d}(\zeta)\ge \zeta,\qquad \forall \zeta\in\mathbb R.
\end{align*}
Then, due to $W_{\rm d}(s_\kappa)\le \overline W_{\rm d}$ and $1+e^{z_{\rm d}(s_\kappa)}>1$, \eqref{z_bound_interval} implies, for all $t\in J_\kappa$,
\begin{align}\label{z_on_Jm}
  z_{\rm d}(t)\le \Phi_{\rm d}\big(z_{\rm d}(s_\kappa)\big).
\end{align}
We now bound $z_{\rm d}(s_0)$ at the beginning of the first OFF interval.

There are two cases:
(i) $\eta(0)\ge 0$. Then $0\in\mathcal I_{\rm off}$ and hence $s_0=0$. In this
case,
\begin{align}\label{zd_s0_case1}
  z_{\rm d}(s_0)=z_{\rm d}(0).
\end{align}

(ii) $\eta(0)<0$. Then $0\notin\mathcal I_{\rm off}$, so $s_0>0$, and we have $[0,s_0)\subseteq \mathcal I_{\rm on}$. Using
\eqref{z_ddot},\eqref{Wd_global_bound},\eqref{Delta_z_active}, we obtain
\begin{align*}
\begin{split}
  z_{\rm d}(s_0)-z_{\rm d}(0)
  &= \int_0^{s_0} \dot z_{\rm d}(s) ds \\
  &\le \int_{\mathcal I_{\rm on}} \dot z_{\rm d}(s) ds
   \le \Delta z_{\rm d}^{\rm on},
\end{split}
\end{align*}
and thus
\begin{align}\label{zd_s0_case2}
  z_{\rm d}(s_0)\le z_{\rm d}(0)+\Delta z_{\rm d}^{\rm on}.
\end{align}
Combining \eqref{zd_s0_case1} and \eqref{zd_s0_case2}, and due to $\Delta z_{\rm d}^{\rm on}\ge 0$, we obtain, in both
cases,
\begin{align}\label{zd_s0_combined}
  z_{\rm d}(s_0)\le \max\big\{z_{\rm d}(0),\, z_{\rm d}(0)+\Delta z_{\rm d}^{\rm on}\big\}
  = z_{\rm d}(0)+\Delta z_{\rm d}^{\rm on}.
\end{align}
Since $\Phi_{\rm d}(\zeta)\ge \zeta$,
\eqref{zd_s0_combined} implies
\begin{align}\label{zd_s0_Psi_bound_corrected}
  z_{\rm d}(s_0)
  \le \Phi_{\rm d}\big(z_{\rm d}(0)\big)+\Delta z_{\rm d}^{\rm on}.
\end{align}
We next proceed by defining
\begin{align}\label{Psi_d_def}
  \Psi_{\rm d}(\zeta)
  := \Phi_{\rm d}(\zeta) + \Delta z_{\rm d}^{\rm on},
\end{align}
which is a continuous, non-decreasing map $\Psi_{\rm d}:\mathbb R\to\mathbb R$,
because $\Phi_{\rm d}$ is non-decreasing in $\zeta$ and $\Delta z_{\rm d}^{\rm on}$ is constant. For $n\in\mathbb N$, we denote by
$\Psi_{\rm d}^{(n)}$ the $n$-fold composition of $\Psi_{\rm d}$ with itself
and set $\Psi_{\rm d}^{(1)}(\zeta):=\Psi_{\rm d}(\zeta)$.

We now show by induction on $\kappa$ that
\begin{align}\label{zd_sm_Psi_iter}
  z_{\rm d}(s_\kappa)\le \Psi_{\rm d}^{(\kappa+1)}\big(z_{\rm d}(0)\big),\qquad \kappa=0,1,\dots,N_{\rm off}.
\end{align}
For $\kappa=0$, \eqref{zd_s0_Psi_bound_corrected} gives
\begin{align*}
  z_{\rm d}(s_0)\le \Phi_{\rm d}\big(z_{\rm d}(0)\big)+\Delta z_{\rm d}^{\rm on}
  = \Psi_{\rm d}\big(z_{\rm d}(0)\big)
  = \Psi_{\rm d}^{(1)}\big(z_{\rm d}(0)\big),
\end{align*}
so \eqref{zd_sm_Psi_iter} holds for $\kappa=0$. Assume that \eqref{zd_sm_Psi_iter} holds for some $\kappa\in\{0,\dots,N_{\rm off}-1\}$.
We show it for $\kappa+1$.

Between $t_\kappa$ and $s_{\kappa+1}$ the filter is ON whenever $t_\kappa<s_{\kappa+1}$, so
\begin{align*}
  z_{\rm d}(s_{\kappa+1}) - z_{\rm d}(t_\kappa)
  = \int_{t_\kappa}^{s_{\kappa+1}} \dot z_{\rm d}(s)\,ds
  \le \int_{\mathcal I_{\rm on}} \dot z_{\rm d}(s)\,ds
  \le \Delta z_{\rm d}^{\rm on},
\end{align*}
and
\begin{align}\label{zd_sm1_from_tm}
  z_{\rm d}(s_{\kappa+1})
  \le z_{\rm d}(t_\kappa)+\Delta z_{\rm d}^{\rm on}.
\end{align}
If $t_\kappa=s_{\kappa+1}$, then the integral above is zero and 
$z_{\rm d}(s_{\kappa+1}) = z_{\rm d}(t_\kappa)$, so \eqref{zd_sm1_from_tm} holds trivially. On the OFF interval $J_\kappa=[s_\kappa,t_\kappa]$, we have from
\eqref{z_on_Jm} that
\begin{align*}
  z_{\rm d}(t_\kappa)\le \Phi_{\rm d}\big(z_{\rm d}(s_\kappa)\big).
\end{align*}
Combining this with \eqref{zd_sm1_from_tm} and recalling \eqref{Psi_d_def} gives
\begin{align*}
  z_{\rm d}(s_{\kappa+1})
  \le \Phi_{\rm d}\big(z_{\rm d}(s_\kappa)\big)+\Delta z_{\rm d}^{\rm on}
  = \Psi_{\rm d}\big(z_{\rm d}(s_\kappa)\big).
\end{align*}
Using the induction hypothesis \eqref{zd_sm_Psi_iter} for $\kappa$ and the
monotonicity of $\Psi_{\rm d}$, we obtain
\begin{align*}
  z_{\rm d}(s_{\kappa+1})
  \le \Psi_{\rm d}\Big(\Psi_{\rm d}^{(\kappa+1)}\big(z_{\rm d}(0)\big)\Big)
  = \Psi_{\rm d}^{(\kappa+2)}\big(z_{\rm d}(0)\big),
\end{align*}
which proves \eqref{zd_sm_Psi_iter} for $\kappa+1$. By induction,
\eqref{zd_sm_Psi_iter} holds for all $\kappa=0,1,\dots,N_{\rm off}$.

In particular, for $\kappa=N_{\rm off}$, we have
\begin{align}\label{zd_sM_Psi}
  z_{\rm d}(s_{N_{\rm off}})\le \Psi_{\rm d}^{(N_{\rm off}+1)}\big(z_{\rm d}(0)\big).
\end{align}

We now use the fact that the last OFF interval $[s_{N_{\rm off}},+\infty)$ is
unbounded on the right and contained in $\mathcal I_{\rm off}$ (see Remark \ref{rmksfz}). On $[s_{N_{\rm off}},+\infty)$
the safety filter is never ON, so the controller coincides with the DADS-BS controller for all $t\ge s_{N_{\rm off}}$. Therefore, similar to \eqref{z_on_Jm}, we obtain using Lemma \ref{krflslm}  
\begin{align}\label{zd_tail_Phi}
  z_{\rm d}(t)\le \Phi_{\rm d}\big(z_{\rm d}(s_{N_{\rm off}})\big),\qquad \forall t\ge s_{N_{\rm off}}.
\end{align}
Combining \eqref{zd_sM_Psi} and \eqref{zd_tail_Phi}, and using the
monotonicity of $\Phi_{\rm d}$, we obtain
\begin{align}\label{zd_tail_PhiPsi}
  z_{\rm d}(t)
  \le \Phi_{\rm d}\Big(\Psi_{\rm d}^{(N_{\rm off}+1)}\big(z_{\rm d}(0)\big)\Big),
  \qquad \forall t\ge s_{N_{\rm off}}.
\end{align}
Finally, since $\dot z_{\rm d}(t)\ge 0$, the function
$z_{\rm d}(\cdot)$ is non-decreasing. In particular, for any $t\in[0,s_{N_{\rm off}}]$,
\begin{align*}
  z_{\rm d}(t)\le z_{\rm d}(s_{N_{\rm off}})
  \le \Phi_{\rm d}\Big(\Psi_{\rm d}^{(N_{\rm off}+1)}\big(z_{\rm d}(0)\big)\Big),
\end{align*}
where we used $\Phi_{\rm d}(\zeta)\ge\zeta$ and \eqref{zd_sM_Psi}. Therefore,
\eqref{zd_tail_PhiPsi} actually holds for all $t\ge 0$:
\begin{align*}
  z_{\rm d}(t)
  \le \Phi_{\rm d}\Big(\Psi_{\rm d}^{(N_{\rm off}+1)}\big(z_{\rm d}(0)\big)\Big).
\end{align*}
From \eqref{M_bound} we have $N_{\rm off}+1\le N_{\boldsymbol{\Omega}}+1$. Furthermore, we have that $\Psi_{\rm d}^{(\kappa_1)}(\zeta)\leq \Psi_{\rm d}^{(\kappa_2)}(\zeta)$ for $\kappa_2>\kappa_1$ and all $\zeta\in\mathbb{R}$ (this can be shown using \eqref{Psi_d_def} and recalling $\Phi_{\rm d}(\zeta) \geq \zeta$ for all $\zeta\in\mathbb{R}$ and $\Delta z_{\rm d}^{\rm on}\geq 0$). Hence
\begin{align*}
  \Psi_{\rm d}^{(N_{\rm off}+1)}\big(z_{\rm d}(0)\big)
  \le \Psi_{\rm d}^{(N_{\boldsymbol{\Omega}}+1)}\big(z_{\rm d}(0)\big),
\end{align*}
and, by the monotonicity of $\Phi_{\rm d}$,
\begin{align}\label{zd_final_bound}
  z_{\rm d}(t)
  \le \Phi_{\rm d}\Big(\Psi_{\rm d}^{(N_{\boldsymbol{\Omega}}+1)}\big(z_{\rm d}(0)\big)\Big),
\end{align}
for all $t\geq 0$. Since $z_{\rm d}(t)\ge z_{\rm d}(0)$ and $z_{\rm d}$ is non-decreasing and
bounded from above by the right-hand side of \eqref{zd_final_bound}, the limit
$\lim_{t\to\infty} z_{\rm d}(t)$ exists and is finite, and satisfies the same
upper bound. Recalling $\Phi_{\rm d}$ and $\Psi_{\rm d}$ are given by \eqref{Phi_d_def} and \eqref{Psi_d_def}, respectively, and that $\Delta z_{\rm d}^{\rm on}$ depends on $\boldsymbol x(0),z_{\rm d}(0),\| v_{\rm g,d}\|_\infty,\| v_{\rm g,q}\|_\infty,T_{\boldsymbol{\Omega}}$, we conclude that the right-hand side of \eqref{zd_final_bound} takes the form of a function $B_{\rm d,\boldsymbol{\Omega}}\big(\boldsymbol{x}(0),z_{\rm d}(0),\Vert v_{\rm g,d}\Vert_{\infty}, \Vert v_{\rm g,q}\Vert_{\infty},N_{\boldsymbol{\Omega},}T_{\boldsymbol{\Omega}}\big)$.  

An analogous analysis for $z_{\rm q}$ can be carried out in the same way. This completes the proof of Theorem \ref{thm6}. \hfill $\qed$

\section{Conclusions}
\label{secCon}
This paper proposed a nonlinear droop-based grid-forming (GFM) controller, by integrating a backstepping (BS) inner--outer architecture with the deadzone-adapted disturbance suppression (DADS) framework. No knowledge of network parameters beyond the point of common coupling (PCC) is required, and the grid voltage is treated as an unknown but bounded disturbance of arbitrarily large magnitude. 
The controller guarantees global well-posedness and boundedness of all closed-loop states (Theorem \ref{thm1}) and PCC voltage-formation performance with assignable residual bounds (Theorem \ref{thm2}). Hard over-current protection is enforced via a minimally invasive control-barrier-function-based safety filter with a closed-form solution. The safety filter is applicable to \textit{any} locally Lipschitz nominal controller, including DADS-BS, and ensures forward invariance of the safe-current set and boundedness of all physical states (Theorem \ref{thm4}), as well as boundedness of the adaptive states in DADS-BS under mild activation assumptions (Theorem \ref{thm6}). Numerical results illustrate improved transient performance and faster recovery during current-limiting events when DADS-BS is used as the nominal controller, compared with conventional PI-based GFM control.


\bibliographystyle{IEEEtran}
\bibliography{main}
\end{document}